\documentclass{article}
\usepackage[utf8]{inputenc}

\usepackage{hyperref}
\usepackage{graphicx}

\usepackage{fullpage, amsmath, amsthm, amssymb}
\usepackage{enumerate}
\usepackage{float}
\usepackage{color}
\usepackage{cases}
\newtheorem{theorem}{Theorem}
\newtheorem{lemma}[theorem]{Lemma}

\newtheorem{claim}[theorem]{Claim}

\newtheoremstyle{restate}{}{}{\itshape}{}{\bfseries}{~(restated).}{.5em}{\thmnote{#3}}
\theoremstyle{restate}

\newcommand{\algwidth}{0.97\textwidth}
\newcommand{\E}{\mathop{\mathbb{E}}}
\newcommand{\eps}{\varepsilon}
\DeclareMathOperator*{\sign}{sign}

\newcommand{\cI}{\mathcal{I}}
\newcommand{\cP}{\mathcal{P}}
\newcommand{\cQ}{\mathcal{Q}}

\newcommand{\cW}{\mathcal{W}}
\newcommand{\cD}{\mathcal{D}}
\newcommand{\cT}{\mathcal{T}}
\newcommand{\cS}{\mathcal{S}}

\newcommand{\bU}{\mathbf{U}}
\newcommand{\bs}{\mathbf{s}}
\newcommand{\bT}{\mathbf{T}}
\newcommand{\bI}{\mathbf{I}}
\newcommand{\bt}{\mathbf{t}}
\newcommand{\adv}{\mathrm{adv}}

\newcommand{\total}{\mathrm{tot}}
\newcommand{\cache}{\mathrm{cac}}

\newcommand{\msg}{{M_A}}
\newcommand{\cont}{\mathrm{cont}}
\newcommand{\addr}{\mathrm{addr}}

\newcommand{\csim}{C_{\mathrm{sim}}}

\renewcommand{\log}{\lg}

\title{Super-Logarithmic Lower Bounds for Dynamic Graph Problems}

\author{Kasper Green Larsen\thanks{\texttt{larsen@cs.au.dk}. Supported by a DFF Sapere Aude Research Leader Grant No. 9064-00068B.}\\Aarhus University \and Huacheng Yu\thanks{\texttt{yuhch123@gmail.com}. Supported by Simons Junior Faculty Award - AWD1007164.}\\Princeton University}

\date{}

\begin{document}

\maketitle

\begin{abstract}
    In this work, we prove a $\tilde{\Omega}(\lg^{3/2} n )$ unconditional lower bound on the maximum of the query time and update time for dynamic data structures supporting reachability queries in $n$-node directed acyclic graphs under edge insertions. This is the first super-logarithmic lower bound for any natural graph problem. In proving the lower bound, we also make novel contributions to the state-of-the-art data structure lower bound techniques that we hope may lead to further progress in proving lower bounds.
\end{abstract}


\section{Introduction}
Graph problems are among the most well-studied topics in algorithms and data structure, with a wealth of new exciting results every year. Just in the recent year, this included a near-linear time algorithm for max-flow~\cite{maxflow} and single-source shortest paths with negative weights~\cite{ssspneg}. The area of fine-grained complexity has provided a large number of complementary conditional lower bounds via reductions from a few carefully chosen conjectured hard problems. This includes lower bounds for graph Diameter and Radius which follow by reduction from All-Pairs-Shortest-Paths (APSP)~\cite{apsp}, or lower bounds for dynamic graph problems, such as Dynamic Reachability, that follow both from the 3SUM conjecture~\cite{multiphase}, the Online Matrix-Vector conjecture~\cite{OMV} and the Strong Exponential Time Hypothesis~\cite{seth}.

However, if we turn to unconditional lower bounds, the situation is much more depressing. The strongest known unconditional lower bound for any natural graph problem, is an $\Omega(\lg n)$ lower bound on the maximum of the update time and query time for the dynamic maintenance of an undirected graph with connectivity queries. This lower bound is due to P{\v a}tra\c{s}cu and Demaine~\cite{PD04a} and dates back to 2004. In the meantime, stronger techniques for proving lower bounds for dynamic data structures have been developed, including a technique by Larsen~\cite{Larsen12a} for proving $\tilde{\Omega}(\lg^2 n)$ lower bounds for dynamic problems with $\Omega(\lg n)$-bit outputs to queries, and a technique for proving $\tilde{\Omega}(\lg^{3/2} n)$ lower bounds for dynamic decision problems ($1$-bit outputs) due to Larsen, Weinstein and Yu~\cite{LWY18}. Unfortunately, none of these techniques have so far been successfully applied to a natural graph problem and the strongest unconditional lower bound remains the $\Omega(\lg n)$ bounds by P{\v a}tra\c{s}cu and Demaine. The main contribution of this work, is to provide the first such $\omega(\lg n)$ lower bound for a graph problem. Concretely, we prove the following
\begin{theorem}
\label{thm:main}
Any data structure for Dynamic Reachability in $n$-node directed acyclic graphs under edge insertions, with worst case update time $t_u$, expected query time $t_q$ and $w$-bit memory cells for $w= \Omega(\lg n)$, must satisfy 
\[
t_q = \Omega\left( \frac{\lg^{3/2} n}{\lg^2(t_u w)} \right).
\]
\end{theorem}
In addition to the $\Omega(\lg n)$ lower bound by P{\v a}tra\c{s}cu and Demaine, the only other known lower bounds for Dynamic Reachability are two \emph{threshold} lower bounds. First, P{\v a}tra\c{s}cu and Thorup~\cite{PT11} showed that for Dynamic Reachability, even in the undirected case, any data structure with update time $t_u = o(\lg n)$ must have query time $t_q = n^{1-o(1)}$. As a complimentary result, Clifford et al.~\cite{CGL15} showed that any data with query time $t_q = o(\lg n/\lg \lg n)$ and space $n \lg^{O(1)} n$ (for graphs with $n \lg^{O(1)} n$ edges) must have an update time of $t_u = n^{1-o(1)}$. These bounds however do not exceed $\Omega(\lg n)$ on the maximum of $t_u$ and $t_q$.

In addition to proving the first $\omega(\lg n)$ lower bound for a natural graph problem, we also make several contributions on a technical level, hopefully paving the way for further $\omega(\lg n)$ lower bounds for dynamic data structure problems. In the following section, we start by giving a high level overview of previous techniques for proving data structure lower bounds. We then give a more technical description of our new approach and argue how we overcome some of the obstacles that have prevented previous works from breaking the $\lg n$ barrier for graph problems.

\subsection{Lower Bounds in the Cell Probe Model}
\label{sec:intropre}
Unconditional lower bounds on the operation time of data structures are proved in the cell probe model of Yao~\cite{Yao81}. In this model, a data structure consists of a random access memory, partitioned into cells of $w$ bits each. 

For \emph{static} data structures, an input to a data structure problem is preprocessed into a memory representation. For space usage $S$, the data structure uses memory cells of integer addresses $[S] = \{0,\dots,S-1\}$ to represent the input. Preprocessing is free of charge and only the space usage is measured.

To answer a query, the data structure is allowed to read, or \emph{probe}, up to $t_q$ memory cells and must announce the answer to the query based on the contents of the probed cells. In this model, computation is free of charge, and the addresses of the cells to probe may be determined as an arbitrary function of the query and the contents of previously probed cells. Formally, this can be modelled by having a decision tree for each query. Each node of a tree is labeled with a cell address in $[S]$ and has $2^w$ children, corresponding to each possible contents of that cell. The leaves of the trees are labeled with the answer to the query. 

When proving lower bounds for static data structures, we study the tradeoff between $S$ and $t_q$. 

For \emph{dynamic} data structures, we also need to support updates to the underlying data. This could, for example, be edge insertions or deletions in a graph. For dynamic data structures, we assume cells have integer addresses in $[2^w] = \{0,\dots,2^w-1\}$. Queries are still answered as in static data structures. To process an update, a data structure may probe up to $t_u$ memory cells. While probing a cell, the data structure may also change the contents of that cell. Similarly to queries, we only count the number of cell probes when stating update time and any computation is free of charge.

\paragraph{Previous Techniques for Static Lower Bounds.}
The current state-of-the-art technique for proving lower bounds for static data structures, is the \emph{cell sampling} technique by Panigrahy et al.~\cite{PTW10} that was later refined in~\cite{Larsen:2012:focs} and used in early work on hashing by Siegel~\cite{siegel}. As it plays a central role in more advanced techniques for dynamic data structures, including ours, we briefly sketch it here. Also, to present the different techniques in a coherent manner, we use the same example data structure problem when discussing them. The example we use is \emph{2d range sum}. In the static version of this problem, the input is a set of $n$ points in 2d having integer coordinates on the $[n] \times [n]$ grid. Each point is assigned a $b$-bit weight. A query is also specified by a point $(x,y) \in [n] \times [n]$ and the goal is to return the sum of the weights assigned to input points $(x',y')$ with $x' \leq x$ and $y' \leq y$.

To prove a lower bound for this problem using the cell sampling technique, we use an \emph{encoding argument}. The idea is to consider a uniform random assignment of weights to a fixed set of $n$ points. Assuming the availability of an efficient data structure for 2d range sum, we now give an encoding and decoding procedure for reconstructing most of the random weights. Intuitively, if we can reconstruct $m$ weights from the encoding, then the encoding length must be at least $m b$ bits. The idea is thus to give an encoding with a short length as a function of $S$ and $t_q$. For this, the cell sampling technique first constructs a data structure on the random input. For a sampling probability $p = (n/(Sw))^{O(1)}$, we randomly sample each memory cell independently with probability $p$. The encoding is then the contents and addresses of the sampled cells, costing an expected $pS(w + \lg S) = O(pSw) = o(n)$ bits. To reconstruct many weights from this encoding, a decoder can now simulate every possible query $(x,y) \in [n] \times [n]$. When simulating a query, the decoder runs the query algorithm, and for every probed cell, checks whether it is in the random sample that was encoded. If so, the simulation can continue and otherwise the decoder discards the query and moves on to the next query point. Note that for any fixed query, the chance that all $t_q$ cells it probes are in the sample is $p^{t_q}$. If $t_q = o(\lg n/\lg(Sw/n))$, this is $n^{-o(1)}$. Since there are $n^2$ queries, the decoder will succeed in recovering the answer to $n^{2-o(1)}$ queries. If these queries are sufficiently different (the rank is $\Omega(n)$ when interpreted as linear combinations over the point weights), then they together reveal $\Omega(nb)$ bits of information about the weights. This is a contradiction since the encoding size was $o(n)$ and thus one concludes $t = \Omega(\lg n/\lg(Sw/n))$. We note that such a rank argument for 2d range sum was first used by P{\v a}tra\c{s}cu~\cite{Patrascu07} in a communication complexity based lower bound proof~\cite{Patrascu07,PT06}.

\paragraph{The Chronogram Technique for Dynamic Lower Bounds.}
The basic idea in most of the central techniques for proving dynamic lower bounds, is to boost a static lower bound by a logarithmic factor. For this reason, these techniques only apply to so-called \emph{decomposable problems}. To explain this, we start by presenting the seminal chronogram technique by Fredman and Saks~\cite{FS89}. The main idea in their technique, is to consider a sequence of $n$ random updates, partitioned into geometrically decreasing sized \emph{epochs} of $n_i=\beta^i$ updates for some $\beta \geq 2$. The first epoch of updates to be processed is epoch number $\lg_\beta n - 1$, then comes epoch $\lg_\beta n - 2$ and so forth, until epoch $1$. Finally, a random query is asked after the epochs of updates have been processed. For our running example of 2d range sum, each epoch of updates would insert $n_i$ points with uniform random weights and the random query is again a point in $[n] \times [n]$. The answer to the query is still the sum of all weights assigned to points $(x',y')$ with $x' \leq x$ and $y' \leq y$. The key property of this data structure problem, is that it is \emph{decomposable}. By this, we mean that if we know all updates of epochs $j \neq i$ as well as the answer to a query $(x,y)$, then we also know the answer to the query if only the updates of epoch $i$ were present. For 2d range sum, this follows simply by subtracting off the contribution to the sum from epochs $j \neq i$. From this observation, we can intuitively use the dynamic data structure to obtain an efficient static data structure for any epoch $i$.

We explain this in more detail. From a dynamic data structure and a sequence of $n$ updates partitioned into epochs, consider the memory cells of the data structure after having processed all the updates. Each cell of the memory is then associated with the epoch in which it was \emph{updated the last time}. Let $C_i$ denote the cells associated with epoch $i$. There are now a few crucial observations. First, for any epoch $i$, if the updates of the epochs are independent, then the cells in $C_j$ for $j>i$ cannot contain any information about epoch $i$. This is because these cells were last written \emph{before} the updates of epoch $i$ arrived. Secondly, the cells in epochs $j<i$ are very few. Concretely, with an update time of $t_u$, we have $|C_j| \leq \beta^j t_u$ so for $\beta = \omega(t_u w)$, we have $\sum_{j<i} |C_j|w = o(\beta^i)$. This means that all cells that are changed in epochs after $i$ contain $o(1)$ bits on average about each of the $\beta^i$ updates of epoch $i$. Intuitively, if the answer to the random query depends a lot on the updates of epoch $i$, then the data structure has to probe $\Omega(1)$ cells from $C_i$ as other epochs contain too little information to correctly answer the query. Since the sets $C_i$ are disjoint, we may sum this lower bound over all epochs to conclude $t_q = \Omega(\lg_\beta n) = \Omega(\lg n/\lg(t_u w))$.

\paragraph{Static Data Structures with Pre-Initialized Memory and a Cache.}
The above approach can be implicitly seen as proving a static lower bound on each epoch. In this work, we make this connection clearer by defining a special form of static data structure that fits the reduction. We define a static cell probe data structure with pre-initialized memory and a cache, as a static data structure that before seeing the input may pre-initialize all memory cells (with addresses in $[2^w]$) to arbitrary contents (independent of the input). Upon receiving its input, it updates up to $S$ of the memory cells and finally it creates a cache of up to $S_\cache$ memory cells. On a query, the data structure probes memory cells and must announce the answer to the query based on the contents of the probed cells and the cache. The cell probed in each step may be an arbitrary function of the cache and all previously probed cells. Compared to a standard static cell probe data structure, it thus has free access to the cache (which does not count towards the query time), plus it may probe into pre-initialized memory. The query time of a data structure with pre-initialized memory and a cache consists of two parameters. We let $t_\total$ denote the total number of cells probed when answering the query, counting probes to both updated memory cells and cells that were not changed when seeing the input. We let $t_q$ denote the number of probes to updated cells.

The chronogram technique of Fredman and Saks can now be seen as a reduction to a static data structure problem with pre-initialized memory and a cache. To see this, consider again the 2d range sum problem. Among the epochs $\lg_\beta n-1, \dots, 1$, there must be an epoch $i$ where the expected number of probes to cells in $C_i$ is $O(t_q/\lg_\beta n)$ simply by disjointness of the sets $C_i$. Focus on such an epoch $i$. We now obtain a static data structure with pre-initialized memory and a cache for the static problem of 2d range sum on $n_i=\beta^i$ points as follows: Before seeing the input $n_i$ points, we pre-initialize memory by running a hard-coded sequence of updates for epochs $j > i$ using the dynamic data structure. As updates are independent across epochs, this can be done without knowing the input of epoch $i$. Upon receiving the input, we interpret it as the updates of epoch $i$ and run the updates on the dynamic data structure. This updates at most $S = t_u n_i$ cells.  Finally, we run a hard-coded sequence of updates for epochs $j < i$ on the dynamic data structure and put all changed cells in the cache. We thus have $S_\cache \leq \sum_{j < i} \beta^j t_u = o(n_i/w)$. To answer a query on the static data structure, we run the query algorithm of the dynamic data structure. Whenever it probes a cell, if that cell was updated during epoch $j < i$, the contents of the cell is in the cache. Otherwise, we simply probe the memory. When the query algorithm of the dynamic data structure has finished, we exploit that we have a decomposable search problem by subtracting off the contributions to the query answer from the hard-coded epochs $j \neq i$. If $t_q$ is the query time of the dynamic data structure, then $t_\total \leq t_q$ for the static data structure. Furthermore, we have $t'_q = O(t_q /\lg_\beta n)$, where $t'_q$ denotes the number of probes the static data structure makes to updated cells.

In light of the above, the chronogram technique now boils down to proving a $t'_q = \Omega(1)$ lower bound on the number of probes to updated cells for a static data structure with pre-initialized memory and a cache when the data structure updates $S = n_i t_u$ cells and has a cache of size $S_\cache = o(n_i/w)$ cells.

\paragraph{Super-Logarithmic Lower Bounds.}
Building on the chronogram technique, Larsen~\cite{Larsen12a} later developed a technique capable of proving lower bounds of $t_q = \Omega((\lg n/\lg(t_u w))^2)$ for dynamic data structures. His approach is very intuitive in light of the just described reduction to static data structures with pre-initialized memory and a cache: simply use the cell sampling technique (for static data structures) to prove a $t_q = \Omega(\lg n/\lg(Sw/n))$ lower bound for static data structures with pre-initialized memory and a cache. Combining this with the reduction above, a dynamic data structure with query time $t_q$ and update time $t_u$ now gives a static data structure with pre-initialized memory that updates $S=n_i t_u$ cells and has a cache of $o(n_i)$ bits for some epoch $i \in \{\lg_\beta n-1,\dots,(1/2)\lg_\beta n\}$. Furthermore, the static data structure probes $O(t_q/\lg_\beta n)$ updated cells. The lower bound from cell sampling now implies $t_q/\lg_\beta n = \Omega(\lg n_i/\lg(S w/n_i))$. Since $n_i \geq \sqrt{n}$ and $\beta = (t_u w)^{O(1)}$, this is $t_q = \Omega((\lg n/\lg(t_u w))^2)$.

At first sight, this seems like a trivial extension. However, one critical step of the cell sampling technique breaks when attempting to prove the static lower bound. Recall that in cell sampling, each memory cell is sampled independently with probability $p=(n/(Sw))^{O(1)}$. To reach a contradiction for 2d range sum, we want to sample $o(n/w)$ of the updated cells and argue that if the number of probes to updated cells is too small, then $\Omega(n)$ bits of information about the weights may be recovered from the sample and the cache (which also has size $o(n)$ bits). If the number of probes to updated cells was $o(\lg n/\lg(Sw/n))$, then there would indeed be $n^{2-o(1)}$ queries in $[n] \times [n]$ for which all their probes to updated cells are in the sample. If a decoding procedure simulates the query algorithm of any such query, it will recover its answer, since whenever it probes a cell that is not in cache and not in the sample, we know its contents is the same as after pre-initializing the memory (which is independent of the input). Now the critical observation is that we cannot detect which queries have all their probes to updated cells in the sample. For normal static data structures, we could simply discard the query if it probes outside the sample, but with a pre-initialized memory, if a cell is not in cache and not in the sample, we have no clue whether it was an updated cell that was not sampled, or if it was merely a pre-initialized cell. A decoding procedure thus has no clue which queries to simulate. Larsen circumvents this issue by requiring that every weight assigned to a point has $b \geq 3 \lg n$ bits. In this way, the encoding procedure can afford to explicitly write down a set of queries for which the simulation succeeds. Since a query costs $2 \lg n$ bits to write down, but it recovers $b = 3 \lg n$ bits of information about the weights, this still yields a contradiction.

Larsen's technique critically requires a large number of output bits such that the answer to a query reveals more bits than it takes to describe the query. Breaking the logarithmic barrier for a dynamic decision problem, i.e. a problem with a $1$-bit output, remained open for another 5 years until the work of Larsen, Weinstein and Yu~\cite{LWY18}. In their work, which is what we expand upon in this paper, they showed that a static data structure with pre-initialized memory and a cache may be used to obtain a one-way protocol for a natural communication game. In this game, Alice receives a random input to the static data structure problem and Bob receives a random query. Alice then sends Bob a message of $o(n)$ bits and Bob must answer the query correctly with probability slightly better than guessing. Larsen, Weinstein and Yu showed that a data structure with total probes $t_\total$ and $t_q$ probes to updated memory cells can be used to obtain a protocol where the chance of Bob outputting the correct answer is at least
\[
1/2 + \exp(- \tilde{O}(\sqrt{t_\total t_q})).
\]
They applied their technique to 2d range parity, which is just 2d range sum with $1$-bit weights and where the answer needs only be reported mod 2. They also showed that any protocol with $o(n)$ bits of communication can only predict the answer to a query with probability $1/2 + n^{-\Omega(1)}$. Using that $t_q = O(t_\total/\lg_\beta n)$ in the reduction from dynamic to static data structures, this finally implies $t_\total/\sqrt{\lg_\beta n} = \tilde{\Omega}(\lg n)$, i.e. a $\lg^{3/2} n$ lower bound for the original dynamic problem.

\subsection{Technical Barriers Overcome}
\label{sec:newtech}
With previous techniques described, we are now ready to give an overview of the technical barriers we overcome to prove our lower bound for Dynamic Reachability. Clearly, Reachability is a decision problem, so the most relevant previous technique is the technique by Larsen, Weinstein and Yu~\cite{LWY18} for proving $\tilde{\Omega}(\lg^{3/2} n)$ lower bounds. However several obstacles prevent an immediate application of their framework. First, all the previous techniques for dynamic lower bounds above require that the problem is decomposable. If we think about Dynamic Reachability and the reduction from dynamic to static data structures with pre-initialized memory and a cache, this would correspond to each epoch inserting $n_i$ edges of a graph $G_i$. After performing the insertions, one Reachability query is asked. For 2d range sum, the answer to a counting query is the sum over the answers on all epochs and thus is decomposable as we can subtract off contributions from epochs $j \neq i$. But if we think of Reachability, whether a node $s$ can reach a node $t$ when having numerous graphs $G_{\lg_\beta n-1},\dots,G_1$ to traverse, the answer to the Reachability query is an OR over the epochs, not a sum. That is, if $s$ can reach $t$ using the edges in graph $G_j$, then the answer to the query is \textbf{Yes} regardless of the other epochs. The problem is thus not decomposable.

To overcome this, we consider distributions over input graphs $G_i$ where the probability that a query pair of nodes $(s,t)$ can reach each other through $G_i$ is small. Concretely it is around $1/\lg n$. This implies that if we zoom in on an epoch $i$, most query pairs of nodes $(s,t)$ cannot reach each other through graphs $G_j$ with $j \neq i$. In some sense, the other epochs do not block the answer to the query and we are back at a decomposable problem.

Attempting the above creates another issue. The technique by Larsen, Weinstein and Yu critically requires that the answer to a query in the static problem is uniform random among $0$ and $1$. Said briefly, if the query answer is $0$ with probability say $1-1/\lg n$, then it is trivial for Bob to predict the query answer with probability much higher than $1/2$. Our second contribution is thus to adapt their communication game to ``biased" problems, where the query answer is not uniform random. We believe this extension, and our new lower bound, is critical for proving future lower bounds for problems that are not decomposable.

Finally, let us remark that even after having adapted the framework of Larsen, Weinstein and Yu, proving the concrete lower bound for Dynamic Reachability is far from trivial and requires numerous ideas and non-trivial problem specific reductions. We present these in later sections.

\section{Data Structures with Pre-Initialized Memory and a Cache}
In this section, we give more details on static data structures with pre-initialized memory and a cache (as discussed in Section~\ref{sec:intropre}).
Such a data structure has memory cells with addresses in $[2^w]$, where each memory cell is pre-initialized to a fixed content independent of the input data.
Cells may have different initialized contents.
Upon receiving the input data $I$, the preprocessing algorithm will update at most $S$ memory cells and write to \emph{a separate cache} of $S_{\cache}$ cells.
We call this set of $S$ memory cells the \emph{updated cells}.
The set of updated cells is allowed to depend on the input data.

Then the data structure must be able to answer queries $Q$ efficiently by a query algorithm.
The query algorithm can access both the memory and the cache.
Probing the cells in cache is free of charge. 
Let $t_{\total}(Q, I)$ be the query time on $(Q, I)$, i.e., the \emph{total} number of memory cells the query algorithm probes when answering query $Q$ on input $I$.
Let $t_q(Q, I)$ be the query time into updated cells on $(Q, I)$, i.e., the number of \emph{updated cells} it probes for query $Q$ on input $I$.
Finally, we let $t_{\total}=\E[t_{\total}(Q, I)]$ and $t_q=\E[t_q(Q, I)]$ be the expected query times.

\paragraph{Data Structure Problems with Weights and One-Way Communication Problems.}
As discussed in Section~\ref{sec:newtech}, we need to extend the technique of Larsen, Weinstein and Yu~\cite{LWY18} to handle problems where the answer to a query is not uniform random. For this, we define data structure problems with weights. 

Let $\cP$ be a data structure problem with input data from $\cI$ and queries from $\cQ$.
$\cP$ is a data structure problem with weights if $\cP$ maps query-data pairs $\cQ\times\cI$ to real numbers $[-1, 1]$.
A data structure for $\cP$ is only required to compute the \emph{sign} of $\cP(Q, I)$ for query $Q$ on input $I$.
Equivalently, one may view $\cP$ as a data structure problem with one-bit output, and $\left|\cP(Q, I)\right|$ as the \emph{weight} of the instance $(Q, I)$.

Given $\cP$ and a (product) distribution $\cD_Q\times\cD_I$ over $\cQ\times\cI$ such that $\E_I[\cP(Q, I)]=0$ for all $Q$, we define the one-way communication problem $G_{\cP}$ as follows: Alice gets a random input $I\sim\cD_I$, Bob gets a random query $Q\sim\cD_Q$, Alice sends one message $\msg$ of $C$ bits to Bob, and their goal is to maximize the weighted \emph{advantage}
\begin{equation}\label{eqn_adv}
	\E_{Q, \msg}\left[\left|\E_{I}[\cP(Q, I)\mid \msg]\right|\right].
\end{equation}
The maximum value of the advantage is denoted by $\adv(G_{\cP}, \cD_Q, \cD_I, C)$.

When $\cP(Q, I)$ only takes values in $\{-1, 1\}$, it is a normal data structure problem with one-bit output.
In this case,~\eqref{eqn_adv} is simply the bias of the output of the query $Q$ given $\msg$: $\adv(G_{\cP}, \cD_Q, \cD_I, C)$ measures how much advantage Bob has over random guessing after seeing a message from Alice.
In general, we measure this bias when the instances may have different weights.
Note that technically, we could also view it as re-weighing the distribution $\cD_Q\times\cD_I$ according to $\left|\cP(Q, I)\right|$ when calculating the bias, but the expected query times are still defined with respect to the ``unweighted'' distribution $\cD_Q\times \cD_I$.

\subsection{Simulation}
Similarly to the work of Larsen, Weinstein and Yu~\cite{LWY18}, we also show that an efficient data structure with pre-initialized memory and a cache may be used to obtain an efficient protocol for the above one-way communication game. Our simulation theorem is the following
\begin{theorem}
\label{thm:simulate}
	Let $\cP: \cQ\times\cI\rightarrow [-1,1]$ be a data structure problem with weights, and $\cD_Q\times\cD_I$ be a distribution over the queries and inputs such that $\E_I[\cP(Q, I)]=0$ for all $Q\in\cQ$ and $\left|\cP(Q, I)\right|\geq \beta$ for some $\beta>0$.
	If there is a data structure $D$ with pre-initialized memory and a cache for $\cP$ such that $D$ has at most $S$ updated cells, $S_{\cache}$ cells in cache, expected query time $t_{\total}$ and expected query time into updated cells $t_q$, then for any $p\in(0,1)$, there is a one-way communication protocol for $G_{\cP}$ such that
	\[
		\adv(G_{\cP}, \cD_Q, \cD_I, (8pS+S_{\cache})\cdot w)\geq 2^{-O(\sqrt{t_{\total}(t_q\log 1/p+\log 1/\beta)}\cdot \log1/p)}-2^{-\Omega(pS)}.
	\]
\end{theorem}
As the proof is heavily inspired by previous work, we defer the proof to Section~\ref{sec:defer}.

\section{Data Structure Problems}
With the necessary framework for obtaining a one-way protocol from a static data structure with pre-initialized memory and a cache, we are now ready to focus on the concrete problem of Dynamic Reachability. In our proof, we consider Reachability queries in a variant of the so-called Butterfly graph. This graph was also used by P{\v a}tra\c{s}cu~\cite{Pat11} in his seminal work that established lower bounds for a host of static data structure problems. However, when constructing a dynamic problem involving multiple Butterfly graphs, we found it more convenient to define an intermediate data structure problem that we name $0$-XOR in Multiple Butterflies. This problem is easier to prove one-way communication lower bounds for. We then show via a reduction that a data structure for Dynamic Reachability may be used to solve $0$-XOR in Multiple Butterflies. We define this new data structure problem in the following.

\paragraph{Butterfly Graphs.}
For a degree $B$ and depth $d$, a Butterfly graph with degree $B$ and depth $d$ has $d+1$ layers of $B^d$ nodes. The nodes at layer $0$ are the sources and the nodes in layer $d+1$ are the sinks. All nodes, except the sinks, have degree $B$. 

To describe the edges, we index the nodes in every layer by consecutive integers in $[B^d]$. For nodes in layer $i$ for $0 \leq i \leq d$, we have $B$ outgoing edges. The edges leaving a node $v^i_j$ go to the $B$ nodes $v^{i+1}_k$ in layer $i+1$ for which the base-$B$ representation of $j$ and $k$ are equal in all digits $c \neq i$ (they may be equal or not in digit $i$). The least significant digit is digit number $0$ and so forth.

Observe that there is a unique path between any source-sink pair $(s,t)$. This path is obtained by writing $s$ and $t$ in base $B$. Then, for each layer $i=0,\dots,d$, starting at the node $s$ in layer $0$, we think of changing digit $i$ of $s$ to the $i$'th digit of $t$. This means that if we are at some node $v^i_j$ in layer $i$, we go to the node $v^{i+1}_k$ where $j$ and $k$ are equal in all digits except possibly the $i$'th in which the digit of $k$ is equal to that of $t$. See Figure~\ref{fig:butterflyQuery} for an example of a Butterfly of degree $2$ and depth $3$.

\paragraph{$0$-XOR in One Butterfly.}
This is a static data structure problem in which the input is a Butterfly of degree $B$ and depth $d$ with an assignment of a $b$-bit string to every edge of the Butterfly. A query is specified by a source-sink pair $(s,t)$ and the answer to the query is $1$ if the XOR of bit strings along the path from $s$ to $t$ is the all-zero bit string and the answer is $0$ otherwise. See Figure~\ref{fig:butterflyQuery} for an illustration.
\begin{figure}[h]
\centering
\includegraphics[width=5cm, keepaspectratio]{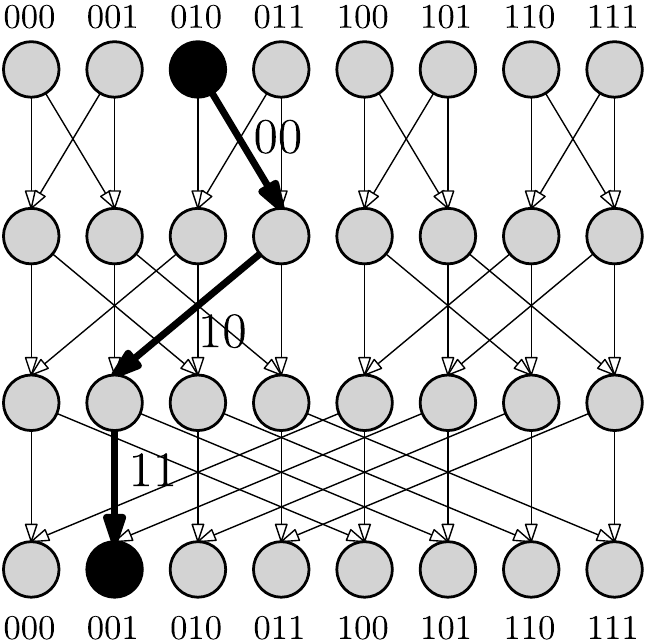}
\caption{A Butterfly with degree $2$ and depth $3$ with $b$-bit strings on the edges for $b=2$. Given the query $(2,1)$, we interpret $2$ as the index of a source and $1$ as the index of a sink. The query $(2,1)$ thus asks whether the XOR of bit strings along the path from the source indexed by $2 = 010$ to the sink indexed by $1 = 001$ is the all-zero bit string. In this example, the XOR is $00 \oplus 10 \oplus 11 = 01$, i.e. not the all-zero bit string. Thus the answer to the query is $0$. For clarity, we have only shown the bit strings on edges along the queried path.}
\label{fig:butterflyQuery}
\end{figure}

\paragraph{$0$-XOR in Multiple Butterflies.}
This is a dynamic data structure problem. For a depth $d$ and degree $B$, we must maintain multiple Butterfly graphs $G_d,\dots,G_1$ all of degree $B$ and where the depth of $G_i$ is $i$. 

Updates in this problem arrive in \emph{epochs}. The first epoch has number $d$, then comes epoch $d-1$ and so forth, until epoch $1$. The updates of epoch $i$ assign a $b$-bit string to each edge of the Butterfly $G_i$.
\begin{figure}[h]
\centering
\includegraphics[width=10cm, keepaspectratio]{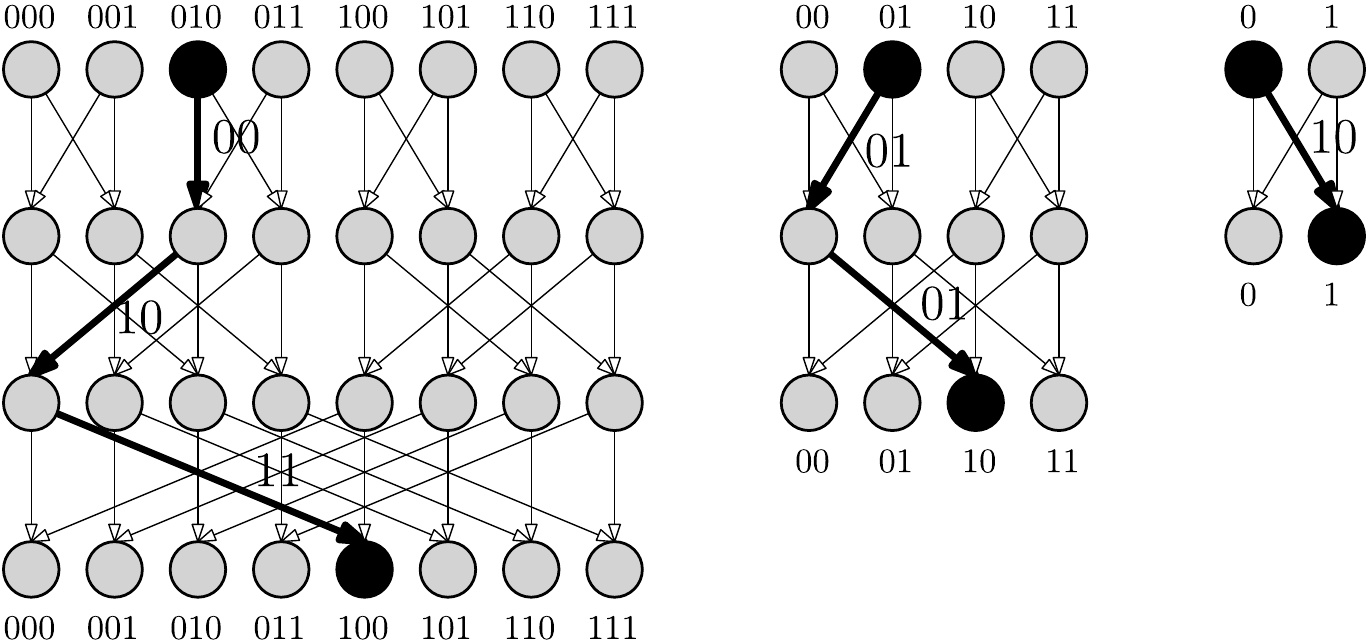}
\caption{Consider $0$-XOR in the three Butterfly graphs of degree $2$ and depths $3,2$ and $1$ respectively, with $b$-bit strings for $b=2$. Given the query $(2,4)$, we interpret $\lfloor 2/2^0 \rfloor = 2 = 010$ as the index of a source in the first graph, $\lfloor 2/2^1 \rfloor = 1 = 01$ as the source index in the second and $\lfloor 2/2^2 \rfloor = 0$ as the index of the source in the third and smallest graph. In the first graph, the sink is $\lfloor 4 /2^0 \rfloor = 4 = 100$, in the second graph, the sink is $\lfloor 4 / 2^1 \rfloor = 2 = 10$ and in the third graph, the sink is $\lfloor 4/2^2 \rfloor = 1$. The query $(2,4)$ thus asks whether the XOR of bit strings along at least one of the three bold paths is the all-zero bit string. This is the case in the second graph where $01 \oplus 01 = 00$ and thus the answer to the query is $1$.}
\label{fig:butterflyMany}
\end{figure}

After the updates have been processed, we must answer queries. A query is specified by two integers $(s,t) \in [B^d] \times [B^d]$. In each Butterfly $G_i$, the pair $(s,t)$ specifies the source-sink pair $(s_i,t_i)$ such that $s_i$ is the source of index $\lfloor s/B^{d-i} \rfloor$ in $G_i$ and $t_i$ is the sink of index $\lfloor t/B^{d-i} \rfloor$ in $G_i$. The answer to a query $(s,t)$ is $1$ if there is \emph{at least one} $G_i$ for which the query $(s_i,t_i)$ has the answer $1$ for $0$-XOR in one Butterfly. Otherwise the answer is $0$. See Figure~\ref{fig:butterflyMany} for an illustration.

\section{$0$-XOR in Multiple Butterflies to Dynamic Reachability}
In this section, we give a reduction from  $0$-XOR in Multiple Butterflies, to Dynamic Reachability (in directed graphs). We perform the reduction in two steps to ease the presentation.

\paragraph{Reduction for One Butterfly.}
Consider first the (static) problem of $0$-XOR in One Butterfly. Recall that in this problem, we are given a Butterfly graph $G$ of degree $B$ and depth $d$ with a $b$-bit string assigned to each edge. For a query $(s,t)$, we must return whether the XOR of the bit strings along the unique $s$-$t$ path is the all-zero bit string $\bar{0}$ or not. 

From $G$, we construct a new graph $G'$ for Reachability queries. The intention is that every $0$-XOR  query in $G$ can be answered by one Reachability query in $G'$. For every node $u$ in $G$, we construct $2^b$ nodes in $G'$. We think of these nodes as representing each of the $2^b$ possible $b$-bit strings. A node $u \in G$ is thus replaced by nodes $\{u^\sigma\}_{\sigma \in \{0,1\}^b}$ in $G'$. For an edge from a node $u$ to a node $v$ in $G$ with bit string $\sigma \in \{0,1\}^b$, we insert $2^b$ edges in $G'$. There is one such edge leaving every node $u^\tau$. The edge leaving $u^\tau$ enters $v^{\tau \oplus \sigma}$, where $\oplus$ denotes bitwise XOR. The key observation is that for any path from a source $s$ to a sink $t$ in $G$, if we start in $s^{\bar{0}}$ in $G'$ and descend along the same path, always following the unique outgoing edge from a current node $u^\tau$ to a node $v^{\tau \oplus \sigma}$ along the path, then there is exactly one reachable sink $t^\sigma$. Moreover, $\sigma$ is equal to the XOR of the bit strings assigned to edges along the $s$-$t$ path in $G$. Thus the answer to a $0$-XOR query $(s,t)$ in $G$ is $1$ if and only if source $s^{\bar{0}}$ can reach sink $t^{\bar{0}}$ in $G'$. See Figure~\ref{fig:butterflyXORtoReach} for an illustration.
\begin{figure}[h]
\centering
\includegraphics[width=7cm, keepaspectratio]{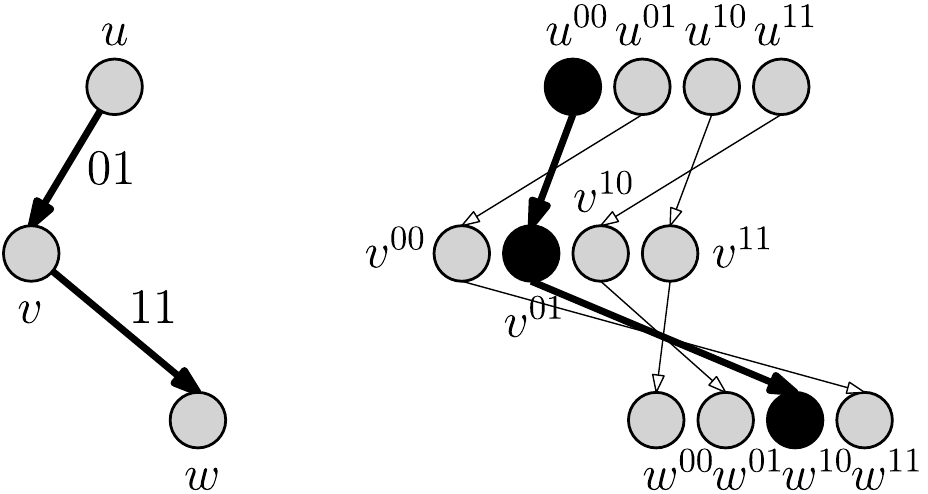}
\caption{On the left, we have three nodes in a Butterfly $G$ with $b$-bit strings on the edges for $b=2$. To transform it into a reachability instance $G'$, each node is represented by $2^b = 4$ nodes, one for each $b$-bit string. For the edge from $u$ to $v$ with bit string $01$, we add an edge from $u^\sigma$ to $v^{\sigma \oplus 01}$ for every $\sigma \in \{0,1\}^2$. The node $u^{\bar{0}} = u^{00}$ can reach precisely the node $w^{10}$ corresponding to the XOR $01 \oplus 11 = 10$.}
\label{fig:butterflyXORtoReach}
\end{figure}

\paragraph{Full Reduction.}
We now use the reduction above for one Butterfly to reduce $0$-XOR in Multiple Butterflies to Dynamic Reachability. In our Dynamic Reachability problem, the graph has $n=\sum_{i=1}^d 2^b (i+1)B^i  + 2\sum_{i=0}^d B^i$ nodes. There are $2^b (i+1) B^i$ nodes corresponding to the Butterfly $G_i$. These nodes correspond to the nodes in $G'_i$ created in the reduction from $0$-XOR in One Butterfly of degree $B$ and depth $i$ to Reachability. The remaining $2 \sum_{i=0}^d B^i$ nodes correspond to two perfect $B$-ary trees $\cS$ and $\cT$ with $B^d$ leaves. Initially, there are no edges in the graph.

Recall that in $0$-XOR in Multiple Butterflies, the updates arrive in epochs, where the updates of epoch $i$ assign bit strings to the edges of $G_i$, where $G_i$ has degree $B$ and depth $i$.

In epoch $d$ (the first to be processed), we start by inserting edges to construct the two perfect $B$-ary trees $\cS$ and $\cT$ with $B^d$ leaves each. We think of the leaves of $\cS$ as representing the sources in $G_d$. The nodes just above represent the sources in $G_{d-1}$ and so forth. Similarly with $\cT$, the leaves represent sinks of $G_d$ and so on. The tree $\cS$ has its edges pointing from children to parents, whereas $\cT$ has its edges pointing from parents to children. For the nodes at depth $i$ in $\cS$ (leaves have depth $d$) we add an edge from the node representing source $s$ in $G_i$ to the corresponding source $s^{\bar{0}}$ in $G'_i$. For the nodes $t$ at depth $i$ in $\cT$, we instead add an edge from the sink $t^{\bar{0}}$ in $G'_i$ to $t$.

To handle the updates of epoch $i$, we simply insert the $2^biB^i$ edges into $G'_i$ corresponding to the reduction shown above from $0$-XOR in One Butterfly to Reachability. Finally, to answer a query $(s,t) \in [B^d] \times [B^d]$ after having processed all epochs of updates, we simply ask whether the leaf corresponding to $s$ in $\cS$ can reach the leaf corresponding to $t$ in $\cT$. To see that this correctly answers the query, observe that the leaf corresponding to $s$ can only leave $\cS$ through one of the ancestors of $s$. If it leaves the ancestor at depth $i$, then it enters the source $s^{\bar{0}}$ indexed $\lfloor s/B^{d-i} \rfloor$ in $G'_i$. To leave $G'_i$, it has to reach a sink $t^{\bar{0}}$. If the sink it reaches is not indexed $\lfloor t/B^{d-i} \rfloor$, then the edge leaving to $\cT$ enters a node $u$ where $t$ is not in the subtree rooted at $u$ and hence cannot reach $t$. On the other hand, if it reaches the sink $t^{\bar{0}}$ indexed $\lfloor t/B^{d-i} \rfloor$, then the leaving edge enters a node $u$ in $\cT$ where $t$ is in the subtree and thus $t$ is reachable. It follows that $s$ can reach $t$ if and only if there is at least one $G'_i$ in which the source $\lfloor s/B^{d-i} \rfloor$ can reach the sink $\lfloor t/B^{d-i} \rfloor$. This concludes the reduction. See Figure~\ref{fig:butterflyReduction} for an illustration. We thus have
\begin{figure}[h]
\centering
\includegraphics[width=10cm, keepaspectratio]{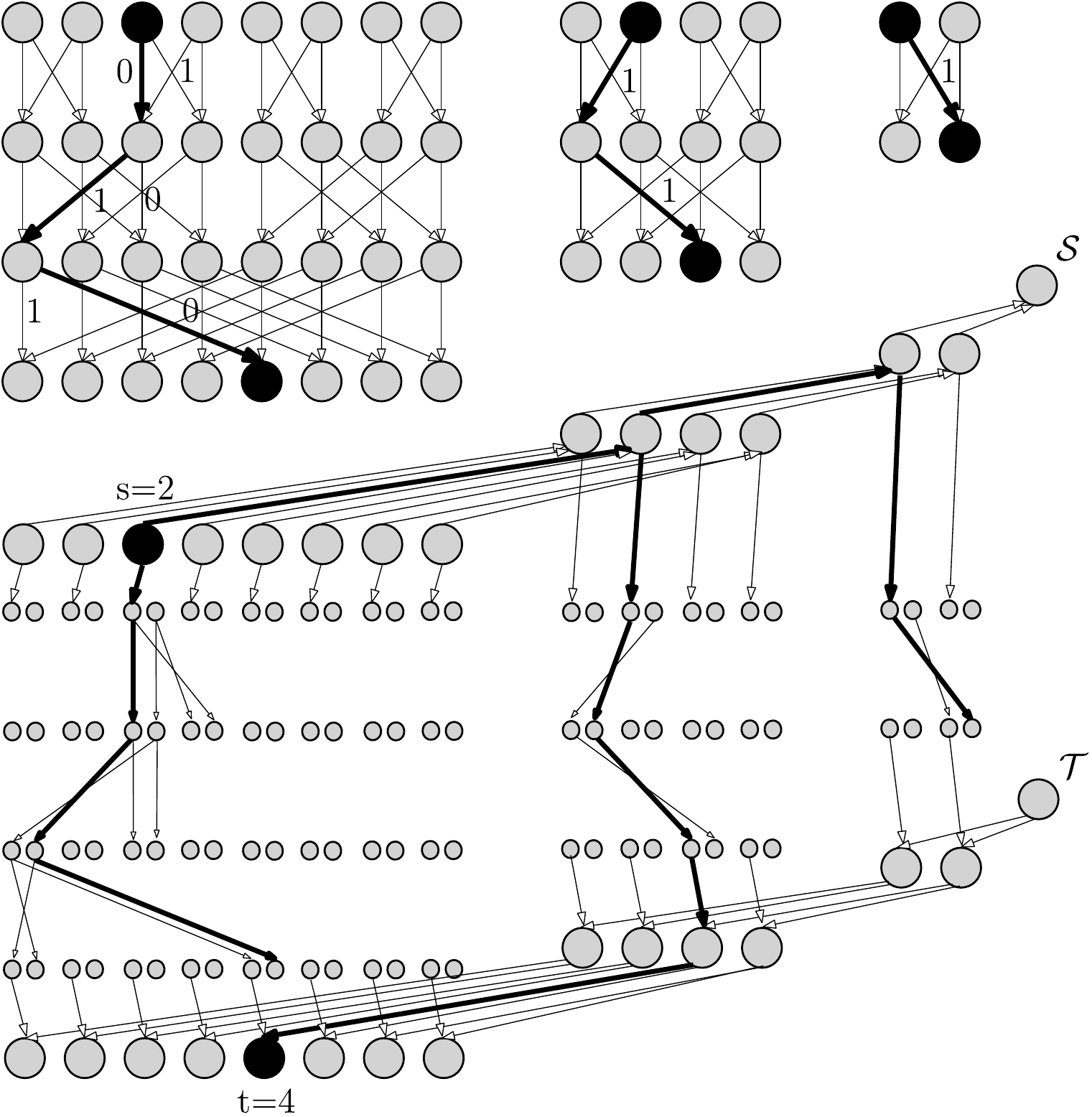}
\caption{Full reduction. The query $(s,t)=(2,4)$ on three Butterfly graphs of degree $2$ and depth $3,2,1$ with $b$-bit strings for $b=1$. For clarity, most edges in the Reachability instance have been hidden. In this example, the answer to the query is $1$. The leaf in $\cS$  corresponding to $s$ can reach the leaf in $\cT$ corresponding to $t$ by using the Butterfly graph of depth $2$ in which the XOR along the path from source $\lfloor s/2 \rfloor$ to sink $\lfloor t/2 \rfloor$ is the all-zero bit string.}
\label{fig:butterflyReduction}
\end{figure}

\begin{theorem}
\label{thm:reduction}
If there is a data structure for Dynamic Reachability in $n$-node graphs, with worst case update time $t_u$ and expected query time $t_q$, then for any $d$, $B$ and $b$ such that $n \geq \sum_{i=1}^d 2^b (i+1)B^i  + 2\sum_{i=0}^d B^i$, there is a data structure for $0$-XOR in Multiple Butterflies of degree $B$ and depth up to $d$ with $b$-bit strings, that probes $O(t_u n/B^{d-i-1})$ cells when processing the updates of epoch $i$, and that answers the query in expected $t_q$ probes.
\end{theorem}
In light of this, we prove a lower bound for $0$-XOR in Multiple Butterflies:
\begin{theorem}
\label{thm:xorlb}
Consider $0$-XOR in Multiple Butterflies of degree $B$ and depth up to $d$ with $(\lg_2  d + 6)$-bit strings, for $d$ at least a sufficiently large constant. If $B = (dt_u w)^{12}$, then any data structure updating $O(t_u d B^{i+1})$ cells in epoch $i$ must have an expected query time $t_q$ satisfying
\[
t_q = \Omega\left(\frac{\lg^{3/2} n}{\lg^{2}(t_u w)}\right)
\]
where $n = \sum_{i=1}^d dB^i$ is the total number of nodes in all Butterfly graphs.
\end{theorem}
Let us combine Theorem~\ref{thm:reduction} and Theorem~\ref{thm:xorlb} to obtain our lower bound for Dynamic Reachability. 

For any sufficiently large $n$, assume we have a data structure for Dynamic Reachability in $n$-node graphs with worst case update time $t_u$ and expected query time $t_q$. Pick $d$ as the largest integer such that for $B = (d t_u w)^{12}$, we have $\sum_{i=1}^d 2^6 d (i+1) B^i + 2 \sum_{i=0}^d B^i \leq n$. Observe that increasing $d$ by one increases the sum by a factor $(t_u w)^{O(1)}$ when $w$ satisfies $w = \Omega(\lg n) = \Omega(d)$. Thus for the chosen $d$, we get $\sum_{i=1}^d 2^6 d (i+1) B^i + 2 \sum_{i=0}^d B^i  \geq n/(t_u w)^{O(1)}$ implying $B^d \geq n/(t_u w)^{O(1)}$.

From Theorem~\ref{thm:reduction} with $b = (\lg_2 d + 6)$, we now obtain a data structure for $0$-XOR in Multiple Butterflies of degree $B = (d t_u w)^{12}$ and depth up to $d$ that probes 
\[
O(t_u n/B^{d-i-1}) = t_u B^d (t_u w)^{O(1)} /B^{d-i-1} = B^{i+1} (t_u w)^{O(1)}
\]
cells when processing the updates of epoch $i$. Theorem~\ref{thm:xorlb}, with $t_u' = (t_u w)^{O(1)}$, finally implies that
\[
t_q = \Omega\left(\frac{\lg^{3/2} n}{\lg^2(t_u w)} \right).
\]
This proves our main result, Theorem~\ref{thm:main}.

\section{Lower Bound for $0$-XOR in Multiple Butterflies}
In this section, we prove our lower bound for $0$-XOR in Multiple Butterflies stated in Theorem~\ref{thm:xorlb}.

The first step of the proof is to reduce $0$-XOR in One Butterfly to $0$-XOR in Multiple Butterflies. Concretely, we show that an efficient (dynamic) data structure for $0$-XOR in Multiple Butterflies gives an (even more) efficient (static) data structure for $0$-XOR in One Butterfly with pre-initialized memory and a cache. The exact statement is as follows
\begin{lemma}
\label{lem:dyntostatic}
If there is a dynamic data structure for $0$-XOR in Multiple Butterflies of degree $B$ and depth up to $d$ with $b$-bit strings for $b \geq \lg_2 d + 6$, with worst case update time $t_u$, and that answers a query in expected $t_q$ probes, then there is an $i \in \{d/2+1,\dots,d\}$, a set of queries $\cQ \subseteq [B^i] \times [B^i]$ with $|Q| \geq B^{2i}/4$ and a deterministic static data structure for $0$-XOR in One Butterfly of degree $B$ and depth $i$, with pre-initialized memory that updates $S=O(t_u d B^{i+1})$ cells and uses a cache of $S_\cache = O(t_u d B^i)$ cells, such that on a uniform random assignment of $b$-bit strings to the edges of the input Butterfly, it holds that
\begin{itemize}
    \item For all queries $q \in Q$, the data structure answers $q$ correctly and the expected number of probes (over the random choice of $b$-bit strings) into updated memory cells is $O(t_q/d)$ and the expected total number of probes $t_\total$ is at most $O(t_q)$.
\end{itemize}
\end{lemma}
The last step is then to prove a lower bound for $0$-XOR in One Butterfly for static data structures with pre-initialized memory and a cache. Here we prove the following
\begin{lemma}
\label{lem:lbstatic}
Consider $0$-XOR in One Butterfly of degree $B$, depth $d$ with $b$-bit strings, for any $B \geq d^7$ and $d$ at least a sufficiently large constant. Assume there is a subset of queries $\cQ \subseteq [B^d] \times [B^d] = [n] \times [n]$ with $|\cQ| \geq n^2/4$ and a deterministic data structure with pre-initialized memory, updating at most $S$ cells, with a cache of size $S_\cache$ making an expected total $t_\total$ probes on any query $(s,t) \in \cQ$ and making at most $t_q$ probes to updated cells in expectation on any query $(s,t) \in \cQ$. Here both expectations are over a uniform random input. If $S_\cache = o(nB^{1/6}/w)$, then it must be the case that
\[
t_\total(t_q \lg(Sw/n) + b) = \Omega((\lg n/\lg(Sw/n))^2).
\]
\end{lemma}
Let us combine the two lemmas to prove Theorem~\ref{thm:xorlb}. Consider $0$-XOR in Multiple Butterflies of degree $B$ and depth up to $d$ with $b=\lg_2 d +6$ bit strings. Assume there is a data structure with worst case update time $t_u$ that answers queries in expected $t_q$ probes. For $B = (d t_u w)^{12}$, we invoke Lemma~\ref{lem:dyntostatic} to obtain a deterministic static data structure for $0$-XOR in One Butterfly $G_i$ of degree $B$ and depth $i \geq d/2+1$ with pre-initialized memory that updates $S=O(t_u d B^{i+1})$ cells and uses a cache of $S_\cache = O(t_u d B^i)$ cells. Furthermore, for $n = B^i$, there is a set $\cQ \subseteq [n] \times [n]$ with $|\cQ| \geq n^2/4$ such that the static data structure answers every query in $\cQ$ in expected total $O(t_q)$ probes and an expected $O(t_q/d)$ probes into updated cells, where the expectation is over a uniform random assignment of $b$-bit strings to the edges of the Butterfly $G_i$.

For $B = (d t_u w)^{12}$, we see that $S_\cache = O(t_u d B^i) = O(t_u d n)=o(nB^{1/6}/w)$. Hence from Lemma~\ref{lem:lbstatic} we get that
\[
t_q((t_q/d) \lg(Sw/n) + \lg d) = \Omega((\lg n/\lg(Sw/n))^2).
\]
Since $\lg n = d \lg B$, we have $\lg d \leq \lg \lg n$. The lower bound thus becomes
\[
t_q = \Omega\left(\min\left\{\frac{\sqrt{d} \lg n}{\lg^{3/2}(Sw/n)} , \frac{\lg^2 n}{\lg^2(Sw/n) \lg \lg n} \right\}\right).
\]
We also have $Sw/n = O(t_u d B) = (t_u w)^{O(1)}$ and $d = \lg_B n$ and thus the lower bound becomes
\[
t_q = \Omega\left(\min\left\{\frac{\lg^{3/2} n}{\lg^{2}(t_u w)} , \frac{\lg^2 n}{\lg^2(t_u w) \lg \lg n} \right\}\right) = \Omega\left(\frac{\lg^{3/2} n}{\lg^{2}(t_u w)}\right).
\]
Since $n = B^i \geq B^{d/2}$ and $B > d$, we have that $\lg n = \Omega(\lg N)$ where $N$ is the total number of nodes in all Butterfly graphs. This concludes the proof of Theorem~\ref{thm:xorlb}.

\subsection{$0$-XOR in One Butterfly to $0$-XOR in Multiple Butterflies}
In the following, we show that an efficient dynamic data structure for $0$-XOR in Multiple Butterflies gives an efficient static data structure for $0$-XOR in One Butterfly. Concretely, we prove Lemma~\ref{lem:dyntostatic} from above.

\begin{proof}[Proof of Lemma~\ref{lem:dyntostatic}]
Recall that in $0$-XOR in Multiple Butterflies, there are $d$ Butterfly graphs $G_d,\dots,G_1$ of depths $d,d-1,\dots,1$. The updates arrive in epochs, starting with epoch $d$, then $d-1$ until epoch $1$. In epoch $i$, we set the bit strings on the edges in $G_i$ independently to uniform random $b$-bit strings. We let $\bU_i$ be the random variable giving the updates of epoch $i$.

After epoch $1$, we ask a query chosen as a uniform random pair $(\bs,\bt) \in [B^d] \times [B^d]$. We let $\cD$ denote the joint distribution of the random updates $\bU = \bU_d,\dots,\bU_1$ and query $(\bs,\bt)$.

The graph $G_i$ has $n_i=(i+1)B^i$ nodes and $m_i = iB^{i+1}$ edges and thus there are $m_i$ updates in epoch $i$.

Given a dynamic data structure for $0$-XOR in Multiple Butterflies, with worst case update time $t_u$ and expected $t_q$ query time to answer the query $(\bs,\bt)$ under the distribution $\cD$, we start by fixing the randomness to obtain a deterministic data structure whose expected query time is no more than $t_q$ under $\cD$.

Next, we zoom in on an epoch $i$ and derive a static data structure for $0$-XOR in One Butterfly. The data structure will have pre-initialized memory corresponding to memory cells changed in epochs $j>i$ and a cache corresponding to cells changed in epochs $j < i$. More formally, for a sequence of updates $U=U_d,\dots,U_1$ in the support of $\cD$, consider processing the updates using the deterministic data structure. Assign each cell in the memory to the epoch in which it was last updated. We let $C_i(U)$ denote the cells assigned to epoch $i$. For a query $(s,t)$, we let $T(U,(s,t))$ denote the set of cells probed when answering $(s,t)$ after updates $U$. By disjointness of the sets $C_i(U)$, we have $t_q = \E[\sum_{i=1}^d |T(\bU,(\bs,\bt)) \cap C_i(\bU)|]$. By linearity of expectation, there must be an epoch $i \in \{d/2+1,\dots,d\}$ for which $\E[|T(\bU,(\bs,\bt)) \cap C_i(\bU)|] \leq 2t_q/d$. Fix such an epoch $i$. 

Next we wish to fix the updates $\bU_{\neq i}=\bU_d,\dots,\bU_{i+1},\bU_{i-1},\dots,\bU_1$ in epochs different from $i$. Intuitively, this will give us a data structure for $0$-XOR in One Butterfly graph corresponding to $G_i$, with the random edge strings given by $\bU_i$. 

For notational convenience, define for a query $(s,t)$ the answer $\phi_j(s,t)$ to $(s,t)$ on $G_j$ as $1$ if the XOR along the path from $s_j = \lfloor s/B^{d-j} \rfloor$ to $t_j = \lfloor t/B^{d-j} \rfloor$ in $G_j$ is the all-$0$ bit string and let $\phi_j(s,t)$ be $0$ otherwise. The answer to the query $(s,t)$ is then $\phi(s,t) = 1-\prod_{j=1}^d(1-\phi_j(s,t))$. Our goal is to fix $\bU_{\neq i}$ such that we can compute $\phi_i(s,t)$ from $\phi(s,t)$ for most pairs $(s,t)$. The main issue is that if any of the epochs $j \neq i$ has $\phi_j(s,t) = 1$, then $\phi(s,t) = 1$ regardless of $\phi_i(s,t)$. In this case, we say that epoch $j$ \emph{blocks} the answer. We will fix $\bU_{\neq i}$ to avoid blocking most queries.

For a query $(s,t)$, call it \emph{valid} for a fixing $U_{\neq i}$ if:
\begin{enumerate}
    \item $(s,t)$ is not blocked by an epoch $j \neq i$.
    \item $\E[|T(\bU,(s,t)) \cap C_i(\bU)| \mid \bU_{\neq i}=U_{\neq i}] \leq 72t_q/d$.
    \item $\E[|T(\bU,(s,t))| \mid \bU_{\neq i} = U_{\neq i}] \leq 36 t_q$. 
\end{enumerate}
For epoch $i$, there are $n=B^i$ sources and $n$ sinks. We claim that there is a fixing $U_{\neq i}$ such that at least $n^2/4$ pairs $(s_j,t_j) \in [n] \times [n]$ satisfies that there is a valid query $(s,t)$ with $s_j = \lfloor s/B^{d-i} \rfloor$ and $t_j = \lfloor t/B^{d-i} \rfloor$. To see this, notice that for a uniform random $\bU_{\neq i}$, any query $(s,t)$ is blocked by epoch $j$ with probability $2^{-b}$. A query $(s,t)$ is thus blocked by an epoch $j \neq i$ with probability at most $d 2^{-b}$. The expected number of blocked queries is thus at most $d 2^{-b}(B^d)^2$. Markov's inequality and a union bound implies the existence of a fixing $U_{\neq i}$ such that at most $3 d 2^{-b}(B^d)^2$ queries are blocked by an epoch $j \neq i$ and at the same time for a uniform $(\bs,\bt)$ we have $\E[|T(\bU,(\bs,\bt)) \cap C_i(\bU)| \mid \bU_{\neq i}=U_{\neq i}] \leq 6t_q/d$ and $\E[|T(\bU,(\bs,\bt))| \mid \bU_{\neq i}=U_{\neq i}] \leq 3t_q$. We claim such a fixing $U_{\neq i}$ satisfies our requirements. 

To see this, we start by defining $\cQ(s_j,t_j)$ as the set of all queries $(s,t) \in [B^d] \times [B^d]$ with $s_j = \lfloor s/B^{d-i} \rfloor$ and $t_j = \lfloor t/B^{d-i} \rfloor$. Now assume for the sake of contradiction that less than $n^2/4$ pairs $(s_j,t_j) \in [n] \times [n]$ satisfies that there is a valid query $(s,t) \in \cQ(s_j,t_j)$. Under this assumption, there are at least $(3/4)n^2$ pairs $(s_j,t_j)$ with no valid query $(s,t) \in \cQ(s_j,t_j)$. Let $(s_j,t_j)$ be an arbitrary such pair. One of the following three must hold: 1) at least $|\cQ(s_j,t_j)|/3$ queries $(s,t) \in \cQ(s_j,t_j)$ are blocked, or 2) at least $|\cQ(s_j,t_j)|/3$ queries $(s,t) \in \cQ(s_j,t_j)$ have $\E[|T(\bU,(s,t)) \cap C_i(\bU)| \mid \bU_{\neq i}=U_{\neq i}] > 72t_q/d$, or 3) at least $|\cQ(s_j,t_j)|/3$ queries $(s,t) \in \cQ(s_t,t_j)$ have $\E[|T(\bU,(s,t)) \cap C_i(\bU)| \mid \bU_{\neq i}=U_{\neq i}] > 36t_q$. Call the first case a type-1 failure, the second case a type-2 failure and the last case a type-3 failure. 

It follows that either there are 1) at least $(1/4)n^2$ pairs $(s_j,t_j)$ that have a type-1 failure, or 2) at least $(1/4)n^2$ pairs $(s_j,t_j)$ that have a type-2 failure or 3) at least $(1/4)n^2$ pairs with a type-3 failure. Since $|\cQ(s_j,t_j)| = (B^{d-i})^2$, the first case gives at least $(1/3)(1/4)(B^{i})^2 (B^{d-i})^2 = (1/12)(B^d)^2$ blocked pairs $(s,t)$. This contradicts that at most $3 d 2^{-b} (B^d)^2$ pairs are blocked if we let $b \geq \lg_2 d + 6$. In the second case, since $(\bs,\bt)$ is uniform random, we have $\E[|T(\bU,(\bs,\bt)) \cap C_i(\bU)| \mid \bU_{\neq i}=U_{\neq i}] > (1/4)(1/3)(72 t_q/d) = 6 t_q/d$. This contradicts that $\E[|T(\bU,(\bs,\bt)) \cap C_i(\bU)| \mid \bU_{\neq i}=U_{\neq i}] \leq 6 t_q/d$. In the third case, we have $\E[|T(\bU,(\bs,\bt)) \cap C_i(\bU)| \mid \bU_{\neq i}=U_{\neq i}] > (1/4)(1/3)(36 t_q) = 3t_q$, which is again a contradiction.

We now have our static data structure. For the fixing $U_{\neq i}$, we pre-initialize the memory by performing all updates in epochs $j > i$. For a uniform random input to $0$-XOR in One Butterfly with degree $B$ and depth $i$, we notice that the distribution of the edge strings is identical to $\bU_i$. We thus think of the input to $0$-XOR in One Butterfly as the updates of epoch $i$ in $0$-XOR in Multiple Butterflies. We thus run the updates in $\bU_i$ to update some of the pre-initialized memory cells. Finally, we run the fixed updates in $U_{\neq i}$ corresponding to epochs $j < i$ and put all cells they change into the cache. We let $\cQ$ be the set of all source-sink pairs $(s_j,t_j)$ for which there is at least one valid query $(s,t)$ in $\cQ(s_j,t_j)$ (this set does not depend on the concrete input $\bU_i$). 

To answer a query $(s_j,t_j) \in \cQ$ for $0$-XOR in One Butterfly, we let $(s,t)$ be the lexicographical first valid query in $\cQ(s_j,t_j)$ and simply execute the query algorithm of the deterministic dynamic data structure on the query $(s,t)$. When it requests a cell, we first check whether it is in the cache (was updated during epochs $j < i$). If so, it is free to access. Otherwise, we simply probe the cell. Since $(s,t)$ is valid, we know it is not blocked and also $\E[|T(\bU,(s,t)) \cap C_i(\bU)| \mid \bU_{\neq i}=U_{\neq i}] \leq 72t_q/d$ and $\E[|T(\bU,(s,t))| \mid \bU_{\neq i}=U_{\neq i}] \leq 36t_q$. Thus $\cQ$ satisfies the claims in the theorem. Since the worst case update time is $t_u$, we have that the updates of epoch $i$ update at most $S = t_u m_i = t_u i B^{i+1} = O(t_u d B^{i+1})$ cells. By the same argument, the cache contains at most $O(\sum_{j<i} t_u d B^{j+1}) = O(t_u d B^i)$ cells.
\end{proof}

\subsection{Lower Bound for $0$-XOR in One Butterfly}
We now turn to proving a lower bound for for the static data structure problem $0$-XOR in One Butterfly when a data structure has pre-initialized memory and a cache. To be compatible with the reduction given in the previous section, the lower bound must hold even if only a constant fraction of the queries can be answered correctly. The goal of the section is to prove Lemma~\ref{lem:lbstatic}.

To prove Lemma~\ref{lem:lbstatic}, we would like to invoke our simulation theorem to obtain a one-way protocol for the communication game in which Alice receives a uniform random input to $0$-XOR in One Butterfly and Bob receives a uniform random query $(\bs,\bt)$ that has good advantage as a function of $S,S_\cache,t_\total$ and $t_q$. Then the lemma would follow by a proving that no such one-way protocol can have a large advantage. 

However, doing this directly on a set of queries $\cQ \subseteq [B^d] \times [B^d]$ will not lead to the desired lower bound. This is because there is a protocol with low communication and high advantage. For ease of notation, let $n = B^d$. To obtain the lower bound we want, we basically need to show that no protocol with communication $n/\lg^{\Theta(1)}n$ can have advantage more than $n^{-\Omega(1)}$. But there is a simple protocol with much higher advantage (at least if $\cQ$ is the set of all queries). We sketch the protocol here for the interested reader:
\begin{itemize}
\item Alice picks the first $n/\lg^{\Theta(1)}n$ sources and sinks. The total number edges on paths between these sources and sinks is $n/\lg^{\Theta(1)}n$ since there is a large overlap in these paths (all leading digits of visited nodes across the layers start with $0$'s). She sends the bit strings on all these edges to Bob. On a uniform random query $(\bs,\bt)$, there is a $1/\lg^{\Theta(1)}n$ chance that both $\bs$ and $\bt$ are among the first $n/\lg^{\Theta(1)}n$ sources and sinks. In this case, Bob knows the answer to the query. If it is not, Bob simply guesses. The advantage is thus $1/\lg^{\Theta(1)}n$, i.e. much higher than $n^{-\Omega(1)}$.
\end{itemize}

What we exploited in the above protocol, is that there is a large collection of queries that together are easy, namely all those $(s,t)$ with $s$ and $t$ among the first $n/\lg^{\Theta(1)}n$ indices. To get around this, we follow previous works and introduce the concept of \emph{meta-queries}. A meta-query is specified by a set of queries $T \subseteq \cQ$ with $|T|=k$, and its answer is the XOR of the answers to all the queries in $T$. Clearly a data structure that can answer every individual query $(s,t) \in T$ in expected $t_\total$ total probes and $t_q$ probes to updated cells, can answer the meta-query in expected $k t_\total$ total probes and $k t_q$ probes to updated cells. We will thus define meta-queries that are sufficiently \emph{well-spread}, i.e. the paths on which to compute XOR's do not share any nodes or edges. We formalize this in the following.

\paragraph{Meta-Queries.}
Let $\cQ \subseteq [n] \times [n]$ be a subset of queries. From $Q$, we construct a set of \emph{meta-queries} $M(\cQ)$. Let $k=n/\lg_2 n$. We add to $M(\cQ)$ one query for every set $T$ of $k$ source-sink pairs $T=\{(s_1,t_t),\dots,(s_k,t_k)\}$ with the property that the unique paths between the source-sink pairs $(s_i,t_i)$ share no nodes. The answer to the query $T=\{(s_1,t_t),\dots,(s_k,t_k)\}$ is obtained as follows: If $I$ is an input to $0$-XOR in One Butterfly, then for a query $(s,t) \in \cQ$, let $\psi_I(s,t)$ be $1$ if the XOR of bit strings on the edges along the path from $s$ to $t$ with edge weights given by $I$ is all-zero. Otherwise, let $\psi_I(s,t) = 0$. Finally, let the answer to the query $T$ be $\bigoplus_{(s,t)\in T} \psi_I(s,t)$. The answer to the meta-query is thus the XOR of the answers to the individual queries.

We first prove that $M(\cQ)$ is large if $\cQ$ is large
\begin{lemma}
\label{lem:largeM}
For any $\cQ \subseteq [n] \times [n]$ with $|\cQ| \geq n^2/4$, if $n$ and $B$ are at least some sufficiently large constants, then we have $|M(\cQ)| \geq n^k/9$.
\end{lemma}

\begin{proof}
We give a probabilistic argument. Consider a fixed meta-query $T \in M(\cQ)$ with $T = \{(s_i,t_i)\}_{i=1}^k$. Now let $T' = (s'_1,t'_1),\dots,(s'_{10k},t'_{10k})$ be an ordered list of $10k$ uniform and independent source-sink pairs. We show two things. First, we argue that there is only a very small chance that all the pairs in $(s_i,t_i) \in T$ are among the queries in $T'$. Next, we argue that there is a large chance that $T'$ contains a subset $T''$ of $k$ queries such that $T'' \in M(\cQ)$. These two can only both be true if $M(\cQ)$ is large.

For the first part, notice that for a fixed $T$, there are at most $\binom{10k}{k}k! (n^2)^{9k}$ choices of $T'$ such that $T \subseteq T'$. The $\binom{10k}{k}$ term accounts for positions in the list $T'$ where $T$ occurs, $k!$ accounts for all permutations of these positions and $(n^2)^{9k}$ accounts for the queries in $T'$ outside the chosen $k$ positions. Since there are $(n^2)^{10k}$ choices for the uniform random $T'$, the probability that $T \subseteq T'$ is then bounded by $\binom{10k}{k}k!n^{-2k} = n^{-2k} (10k)!/(9k)! \leq n^{-2k} (10k)^k = (10k/n^2)^k = (10/( n \lg_2 n))^{k}$. For $n \geq 2^{10}$, this is at most $n^{-k}$.

For the second part, define an indicator $X_i$ for every pair $(s'_i,t'_i)$ in $T'$. The indicator takes the value $1$ if $(s'_i,t'_i) \in \cQ$ and the unique path from $s'_i$ to $t'_i$ in the Butterfly does not share any node with a path between any other pair $(s'_j,t'_j)$ in $T'$. We first argue that $\E[X_i]$ is large. For this, notice first that $\Pr[(s'_i,t'_i) \in \cQ] \geq 1/4$. Next, consider a pair $(s'_j,t'_j)$ with $i \neq j$. For any layer $\ell$ of the Butterfly, the distribution of the node visited by the two paths is uniform random, and the two are independent. Hence the probability that they use the same node in any of the $d+1$ layers is no more than $(d+1)/B^d = (d+1)/n = (\lg_B n)/n$. A union bound over all $10k-1$ choices of $j'$ implies that the path from $s'_i$ to $t'_i$ intersects any of the other paths is no more than $10k (\lg_B n)/n = 10/\lg_2 B$. For $B$ larger than some constant, this is at most $1/100$. It follows that $\Pr[X_i =1] \geq 1/4 - 1/100 = 24/100$. Hence $\E[\sum_i X_i] \geq 10k (24/100) = (24/10)k > 2k$. Now consider the random variable $Y=10k - \sum_i X_i$. This is a non-negative random variable with expectation at most $8k$. Hence by Markov's inequality, we have $\Pr[Y > 9k] < 8/9$. That is, $\Pr[\sum_i X_i < k] < 8/9$. But choosing any set of $k$ pairs $(s'_i,t'_i)$ where $X_i=1$ for all of them results in a set $T'' \in M(\cQ)$. Thus a $T''$ exists with probability at least $1/9$.

We now combine the above two to lower bound $|M(\cQ)|$. For this, notice that the probability that the random $T'$ contains a $T'' \in M(\cQ)$ and yet no $T \in M(\cQ)$ is contained in $T'$ is $0$ (these are contradictory). But a union bound implies that the probability of this event is at least $1/9 - |M(\cQ)|n^{-k}$. Hence we must have $|M(\cQ)| \geq n^{k}/9$.
\end{proof}

\paragraph{Meta-Query Communication Game.}
With meta-queries defined, we are ready to reduce to a one-way communication game. Assume we have a data structure for $0$-XOR in One Butterfly as in the requirements of Lemma~\ref{lem:lbstatic}. Let $\cQ \subseteq [n] \times [n]$ be the promised set of at least $n^2/4$ queries and let $M(\cQ)$ be the corresponding meta-queries. Let $\cI$ be the set of all possible assignments of $b$-bit strings to a Butterfly of degree $B$ and depth $d$.

We now define a data structure problem $\cP$ with weights $M(\cQ) \times \cI \to [-1,1]$. For a query $(s,t) \in M(\cQ)$, define $\phi_I(s,t)$ to be $-(1-2^{-b})$ if $\psi_I(s,t) = 1$ and define it to be $2^{-b}$ otherwise. Now define
\[
\cP(T,I) := \prod_{(s,t) \in T} \phi_I(s,t)
\]
Let us make a few observations about $\cP$. First, we always have $|\cP(T,I)| \geq 2^{-kb}$. Next, notice that if a data structure can compute the answer $\bigoplus_{(s,t) \in T} \psi_T(s,t)$ to the meta-query $T$ on $M(\cQ)$, it can also compute the sign of $\cP(T,I)$. This is because the sign of $\cP(T,I)$ is equal to $(-1)^{\sum_{(s,t) \in T} \psi_I(s,t)}$ and this is determined from the parity of $\sum_{(s,t) \in T} \psi_I(s,t)$ which is equal to $\bigoplus_{(s,t) \in T} \psi_I(s,t)$. 

Let $\cD_{M(\cQ)} \times \cD_\cI$ be the product distribution over $M(\cQ) \times \cI$, giving a uniform random $\bT \in M(\cQ)$ and a uniform random $\bI$ in $\cI$. Observe that for every $T \in M(\cQ)$, it holds that $\E_\bI[\cP(T,\bI)] = 0$. This is because the paths in $T$ are disjoint and hence each of them XOR's to the all-zero bit string $\bar{0}$ with probability precisely $2^{-b}$ and this is independent across all of them. We thus have 
\begin{eqnarray*}
    \E_\bI[\cP(T,\bI)] &=& \E_\bI\left[ \prod_{(s,t) \in T} \left(1_{\psi_\bI(s,t)=1} \cdot (-(1-2^{-b})) + 1_{\psi_\bI(s,t)=0} \cdot 2^{-b}\right)\right] \\
    &=& 
     \prod_{(s,t) \in T} \E_\bI\left[1_{\psi_\bI(s,t)=1} \cdot (-(1-2^{-b})) + 1_{\psi_\bI(s,t)=0} \cdot 2^{-b}\right] \\
     &=&
     \prod_{(s,t) \in T} \left(2^{-b}\cdot (-(1-2^{-b})) + (1-2^{-b})\cdot 2^{-b}\right) \\
     &=& 0.
\end{eqnarray*}
Now observe that the data structure satisfying the assumptions in Lemma~\ref{lem:lbstatic} can answer any meta-query $T \in M(\cQ)$ in a expected total probes $k t_\total$ with an expected $k t_q$ probes to updated cells, where the expectation is over a random $\bI \sim \cD_\cI$. This is because each of the individual queries in $T$ can be answered in expected $t_\total$ total probes and $t_q$ probes into updated cells.

We now wish to invoke Theorem~\ref{thm:simulate} to obtain a one-way protocol for $G_\cP$. We see that we can choose $\beta = 2^{-kb}$. For any $p \in (0,1)$, Theorem~\ref{thm:simulate} now gives a one-way protocol for $G_\cP$ under the product distribution $\cD_{M(\cQ)} \times \cD_{\cI}$ with
\begin{eqnarray*}
    \adv(G_\cP,\cD_{M(\cQ)}, \cD_{\cI}, (8pS + S_\cache) \cdot w) &\geq& 2^{-O(\sqrt{k t_\total (k t_q \lg 1/p + kb)} \cdot \lg 1/p)} - 2^{-\Omega(p S)} \\
    &=& 2^{-O(k \sqrt{ t_\total (t_q \lg 1/p + b)} \cdot \lg 1/p)} - 2^{-\Omega(p S)}.
\end{eqnarray*}
What remains is to prove that no protocol can have a large advantage. This is the contents of the following lemma
\begin{lemma}
\label{lem:commLB}
There is a universal constant $c>0$, such that for $B \geq d^7$ and $d$ at least a sufficiently large constant, 
    any protocol for $G_\cP$ must either send at least $cn B^{1/6}$ bits or have advantage at most $\exp(-c k \lg n)$ under $\cD_{M(\cQ)} \times \cD_\cI$.
\end{lemma}
Before proving Lemma~\ref{lem:commLB}, let us use it to complete the proof of Lemma~\ref{lem:lbstatic}. We choose $p = c'n/(Sw)$ for a sufficiently large constant $c' > 0$. In this case, the $8pSw$ communication is less than $cn/2$ and we conclude that either $S_\cache = \Omega(nB^{1/6}/w)$ or
\[
k \sqrt{ t_\total (t_q \lg(Sw/n) + b)} \cdot \lg(Sw/n) = \Omega(k \lg n + n/w) = \Omega(k \lg n).
\]
This yields the lower bound
\[
t_\total (t_q \lg(Sw/n) + b) = \Omega((\lg n/\lg(Sw/n))^2).
\]
as claimed in Lemma~\ref{lem:lbstatic}. 

What remains is thus to prove the communication lower bound.

\begin{proof}[Proof of Lemma~\ref{lem:commLB}]
Let $\cQ \subseteq [n] \times [n]$ have $|\cQ| \geq n^2/4$, let $M(\cQ)$ be the corresponding meta-queries and $\cP : M(\cQ) \times \cI \to [-1,1]$ be the data structure problem with weights, where $\cI$ is the set of all assignments of $b$-bit strings to the edges of a degree $B$ and depth $d$ Butterfly with $n = B^d$. Let $\cD_{M(\cQ)} \times \cD_\cI$ be the product distribution over a uniform random $\bT$ in $M(\cQ)$ and a uniform random $\bI \in \cI$.

Assume there is a one-way protocol $\pi : \cI \to \{0,1\}^m$ for the communication game $G_\cP$ with advantage $\eps$ and $m$ bits of communication. For any message $A$ that Alice may send, let $\chi_{A} \in \{-1,1\}^{|M(\cQ)|}$ be the vector having one entry per $T \in M(\cQ)$. The entry $\chi_{A}(T)$ takes the value $\sign(\E_{\bI}[\cP(T,\bI) \mid \pi(\bI) = A])$. Similarly, define the vectors $P_I \in [-1,1]^{|M(\cQ)|}$ where the entry corresponding to a $T \in M(\cQ)$ takes the value $\cP(T,I)$.

By definition of the advantage, we have
\[
\eps = |M(\cQ)|^{-1} \E_{\bI}[\langle P_\bI, \chi_{\pi(\bI)}\rangle].
\]
We start by restricting our attention to a concrete message $A$ and corresponding vector $\chi_A$. We claim there must be message $A$ such that both of the following hold
\begin{itemize}
\item $\E_{\bI}[\langle P_\bI, \chi_{\pi(\bI)}\rangle \mid \pi(\bI)=A] \geq (\eps/2)|M(\cQ)|$.
\item $\Pr_{\bI}[\pi(\bI) = A] \geq (\eps/2)2^{-m}$.
\end{itemize}
To see this, assume for the sake of contradiction that no such message exists. Then we would have
\begin{eqnarray*}
    |M(\cQ)|\eps &=& \E_{\bI}[\langle P_\bI, \chi_{\pi(\bI)}\rangle] \\
    &=& \sum_{A \in \{0,1\}^m} \Pr_\bI[\pi(\bI)=A]\E_{\bI}[\langle P_\bI, \chi_{\pi(\bI)}\rangle \mid \pi(\bI)=A] \\
    &=&\sum_{A \in \{0,1\}^m : \Pr_{\bI}[\pi(\bI) = A] < (\eps/2)2^{-m}} \Pr_\bI[\pi(\bI)=A]\E_{\bI}[\langle P_\bI, \chi_{\pi(\bI)}\rangle \mid \pi(\bI)=A]\\ 
    &+& \sum_{A \in \{0,1\}^m : \Pr_{\bI}[\pi(\bI) = A] \geq (\eps/2)2^{-m}} \Pr_\bI[\pi(\bI)=A]\E_{\bI}[\langle P_\bI, \chi_{\pi(\bI)}\rangle \mid \pi(\bI)=A] \\
    &<& \sum_{A \in \{0,1\}^m : \Pr_{\bI}[\pi(\bI) = A] < (\eps/2)2^{-m}} (\eps/2)2^{-m} \E_\bI[\|P_\bI\|_1 \mid \pi(\bI)=A] \\
    &+&\sum_{A \in \{0,1\}^m : \Pr_{\bI}[\pi(\bI) = A] \geq (\eps/2)2^{-m}} \Pr_\bI[\pi(\bI)=A](\eps/2)|M(\cQ)|\\
    &\leq& 2^m (\eps/2) 2^{-m}|M(\cQ)| + (\eps/2)|M(\cQ)|,
\end{eqnarray*}
i.e. a contradiction. So fix such a message $A$ and let $\cI' \subseteq \cI$ be the subset of inputs $I$ for which $\pi(I) = A$. Then $|\cI'| \geq (\eps/2)2^{-m}|\cI|$ and for a uniform random $\bI'$ in $\cI'$, we have $\E_{\bI'}[\langle P_{\bI'} , \chi_A \rangle] \geq (\eps/2)|M(\cQ)|$.

We show that for any $\chi \in \{-1,1\}^{|M(\cQ)|}$, 
\[
\E_{\bI'}[\langle P_{\bI'}, \chi\rangle ]
\]
must be small if $\cI'$ is large. For this, consider the following $r$'th moment, for an even integer $r \geq 2$ to be determined. Here the expectation is over a uniform $\bI$ from $\cI$:
\begin{eqnarray*}
\E_{\bI}[\langle P_\bI, \chi \rangle^r] &=& \\
\sum_{T_1,\dots,T_r} \prod_{i=1}^t \chi(T_i) \cdot \E_{\bI} \left[\prod_{i=1}^r P_{\bI}(T_i) \right] &\leq& \\
\sum_{T_1,\dots,T_r} \left| \E_{\bI} \left[\prod_{i=1}^r P_{\bI}(T_i)\right] \right| &=& \\
\sum_{T_1,\dots,T_r} \left| \E_{\bI} \left[\prod_{i=1}^r \prod_{(s,t) \in T_i} \phi_\bI(s,t) \right] \right|.
\end{eqnarray*}
Now consider a term corresponding to some $T_1,\dots,T_r$. Each of the queries $(s,t) \in T_1 \cup \cdots \cup T_r$ specifies a source-sink path. Assume that there is some edge $e$ that occurs in \emph{exactly} one of these paths and let $(s^\star,t^\star)$ be the corresponding query. Then $\phi_\bI(s^\star,t^\star)$ is independent of all other $\phi_\bI(s,t)$ since even conditioned on the bit strings assigned to all edges other than $e$, the XOR along the $s$-$t$ path is still uniform random. Thus for such $T_1,\dots,T_r$ we have
\begin{eqnarray*}
\E_{\bI} \left[\prod_{i=1}^r \prod_{(s,t) \in T_i} \phi_\bI(s,t) \right]  &=& \\
\E_{\bI}[\phi_{\bI}(s^\star,t^\star)] \E_{\bI} \left[\prod_{i=1}^r \prod_{(s,t) \in T_i \setminus \{(s^\star,t^\star)\}} \phi_\bI(s,t) \right]  &=& 
0.
\end{eqnarray*}
If on the other hand it holds that every edge that occurs in an $s$-$t$ path occurs at least twice, then
\[
\left|\E_{\bI} \left[\prod_{i=1}^r \prod_{(s,t) \in T_i} \phi_\bI(s,t) \right]\right| \leq 1.
\]
Thus if $\Gamma$ denotes the number of lists $T_1,\dots,T_r \in M(\cQ)^r$ for which every edge $e$ in $G_i$ occurs in either zero or at least two $s$-$t$ query paths among all queries in the meta-queries $T_1,\dots,T_r$, then
\[
\E_{\bI}[\langle P_\bI, \chi \rangle^r] \leq \Gamma.
\]
Using that $\Pr_{\bI}[\bI \in \cI'] = |\cI'|/|\cI|$ and non-negativity of $\langle P_I,\chi \rangle^r$, we further have for $\bI'$ uniform in $\cI'$ that
\[
\E_{\bI'}[\langle P_\bI, \chi \rangle^r ] = \E_{\bI}[\langle P_\bI, \chi \rangle^r \mid \bI \in \cI'] \leq \Gamma |\cI|/|\cI'|.
\]
Hence it must be the case that
\[
\E_{\bI'}[\langle P_{\bI'}, \chi \rangle  ] \leq \Gamma^{1/r} (|\cI|/|\cI'|)^{1/r}.
\]
But we showed that there exists a $\chi_A$ with $\E_{\bI'}[\langle P_{\bI'}, \chi_A\rangle] \geq (\eps/2)|M(\cQ)|$ while $|\cI'| \geq (\eps/2)2^{-m} |\cI|$. Thus we conclude that
\[
(\eps/2)|M(\cQ)| \leq \Gamma^{1/r} (\eps^{-1} 2^{m+1})^{1/r}.
\]
For $r \geq 8$, this implies that
\[
(\eps/2)^{1+1/r} \leq |M(\cQ)|^{-1} \Gamma^{1/r} 2^{m/r} \Rightarrow \eps \leq 2 |M(\cQ)|^{-8/9} \Gamma^{1/r} 2^{m/r}.
\]
By Lemma~\ref{lem:largeM}, we have $|M(\cQ)| \geq n^k/9$ and thus
\[
\eps \leq 18 n^{-8k/9}\Gamma^{1/r} 2^{m/r}.
\]
What remains is thus to choose $r$ and to upper bound $\Gamma$. For this, notice that any $T_1,\dots,T_r$ where every edge occurs either zero or at least two times along the query paths can be uniquely described as follows: First, specify a set of edges $E$ such that all query paths use only edges in $E$. Since there are $r k$ query paths of length $d$, and every edge occurs either zero or at least two times, there is a such a set $E$ of cardinality only $rkd/2$. With such an edge set specified, each $T_i$ is uniquely determined by specifying its $kd$ edges as a subset of $E$. This is because all query paths in one meta-query share no nodes or edges. Thus knowing the set of edges occurring in the paths uniquely determine the whole set of queries in $T_i$. This is where we exploit that the meta-queries are well-spread. Recalling that $k=n/\lg_2 n$ and $n=B^d$, we have thus argued that
\[
\Gamma \leq \binom{dB^{d+1}}{rkd/2} \binom{rkd/2}{kd}^r \leq \left(\frac{2e B^{d+1}}{rk} \right)^{rkd/2} \left(\frac{er}{2}\right)^{rkd} = \left(\frac{e^3 r B^{d+1}}{2k} \right)^{rkd/2} = \left(\frac{e^3 r B \lg_2 n }{2} \right)^{rkd/2}
\]
Now fix $r = B^{1/6}$ and assume $B \geq d^7$ and $d$ at least a sufficiently large constant. For such $B$ and $d$, it holds that $(e^3 \lg_2 n) = (e^3 d \lg_2 B) \leq B^{1/3}$. We then have
\[
\Gamma \leq (B^{3/2})^{r k d/2} = B^{(3/4) r k \lg_B n} = 2^{(3/4)r k \lg_2 n} = n^{(3/4) r k}.
\]
Inserting this above, we conclude
\[
\eps \leq 18 n^{-8k/9} n^{(3/4) k} 2^{m/r} \leq 18 n^{-k/9} 2^{m/r} \leq 18 n^{-k/9} 2^{m/B^{1/6}}.
\]
Since $k = n/\lg_2 n$, we conclude that there is a constant $c>0$ such that either $m \geq c n B^{1/6}$ or $\eps \leq \exp(-c k \lg n)$.
\end{proof}

\section{Deferred Proof of Simulation Theorem}
\label{sec:defer}
In this section, we give the deferred proof of Theorem~\ref{thm:simulate}. In the proof, we make use of the Peak-to-Average lemma by Larsen, Weinstein and Yu:
\begin{lemma}[Peak-to-Average~\cite{LWY18}]\label{lem_peak_to_avg}
	Let $f:\Sigma^k\rightarrow \mathbb{R}$ be any real function on the length-$k$ strings over alphabet $\Sigma$, satisfying:
	\begin{enumerate}
		\item $\sum_{z\in\Sigma^k}\left|f(z)\right|\leq 1$; and 
		\item $\max_{z\in\Sigma^k}\left|f(z)\right|\geq \varepsilon$
	\end{enumerate}
	for some $\varepsilon\in(0,1]$.
	Then there exists a subset $Y$ of indices, $\left|Y\right|\leq O(\sqrt{k\cdot\lg(1/\varepsilon)})$, such that
	\[
		\sum_{y\in\Sigma^Y}\left|\sum_{z\mid_Y=y}f(z)\right|\geq \exp(-\sqrt{k\cdot\lg(1/\varepsilon)}).
	\]
\end{lemma}
Note that the lemma as stated in~\cite{LWY18} requires $\varepsilon\geq 2^{-O(k)}$.
However, the statement trivially holds when $\varepsilon\leq 2^k$, since we could simply let $Y=[k]$ and we have
\[
	\sum_{y\in\Sigma^Y}\left|\sum_{z|_Y=y} f(z)\right|=\sum_z \left|f(z)\right|\geq \varepsilon.
\]

We are ready to give the proof of Theorem~\ref{thm:simulate}
\begin{proof}[Proof of Theorem~\ref{thm:simulate}]
	By Markov's inequality and union bound, with probability $1/2$ over a random $Q\in\cD_Q$, we have
	\[
		\E_{I}[t_{\total}(Q, I)]\leq 4t_{\total} \qquad\textrm{and}\qquad \E_{I}[t_q(Q, I)]\leq 4t_q.
	\]
	Denote this set of queries by $\cQ'$.
	Since \eqref{eqn_adv} is an expectation of a nonnegative term, it suffices to prove the lower bound for $Q\in\cQ'$.
	Thus, we may assume without loss of generality for every $Q$, $t_{\total}$ and $t_q$ are upper bounds on the expected times over a random input, and this only loses constant factors.

	To construct a protocol with large advantage, we use a similar strategy to~\cite{LWY18} to simulate data structure $D$.
	We first let Alice simulate the preprocessing algorithm of $D$ on $I$.
	Then Alice samples every \emph{updated} cells with probability $p$, and sends Bob the sampled cells, denoted by $C_0$.
	Then Alice uses public randomness to sample every cell with probability $p$, and sends Bob among the sampled cells, which ones are updated and their contents.
	Denote this set of sampled cells by $C_1$.
	Finally, Alice sends Bob all cells in cache, denoted by $C_2$.
	The triple $(C_0,C_1,C_2)$ forms the message.
	See Figure~\ref{figure:pi} for a formal description of the protocol.

	\begin{figure}
	\begin{tabular}{|l|}
	\hline
	\begin{minipage}{\algwidth}
	\vspace{1ex}
	\begin{center}
	\textbf{One-way protocol $\pi$ for $G_\cP$}
	\end{center}
	\vspace{0.5ex}
	\end{minipage}\\
	\hline
	\begin{minipage}{\algwidth}
	\vspace{1ex}

	By ``sending a cell'', we mean sending the address and content of the cell.
	We use $C\leftarrow z$ to denote that memory cells $C$ have content $z$.
	\begin{enumerate}
	    \item Alice generates the memory state and cache of $D$ by simulating the data structure on $I$.
	    \item\label{step_c0} Alice samples each \emph{updated} cell independently with probability $p$. 
	    Let $C_0$ be the set of sampled cells.
	    If $|C_0|\geq 2pS$, Alice sends a bit $0$ and aborts.
	    Otherwise, she sends a bit $1$, followed by all cells in $C_0$. 
	    \item\label{step_c1} Alice uses \emph{public randomness} to sample \emph{every cell} independently with probability $p$. 
	    Let $C_1$ be the set of sampled cells.
	    If there are at least $2pS$ \emph{updated} cells in $C_1$
	    Alice sends a bit $0$ and aborts.
	    Otherwise, she sends a bit $1$, followed by all \emph{updated} cells in $C_1$.
	    \item\label{step_c2} Alice sends Bob all cells \emph{in cache}.
	    Denote this set of cells by $C_2$.
	\end{enumerate}
	\paragraph{(For the purpose of analysis only:)}
	\begin{enumerate}
	    \setcounter{enumi}{4}
	    \item\label{step_bob} Bob generates the pre-initialized memory $M_0$ of $D$. 
	    Bob updates the contents of $C_0$ in $M_0$ and updates the cache $C_2$, obtain a memory state $M'$, and then simulates the query algorithm of $D$ on query $Q$ with memory state $M'$.
		Let $\csim$ be the set of (memory addresses of) cells probed by $D$ in this simulation. 
		\item\label{step_pta} 
	    Apply Lemma~\ref{lem_peak_to_avg} with $\Sigma:=[2^w]$, $k:= |\csim|, \varepsilon := \beta\cdot p^{4t_q}/4$, and
	    \[
	    	f(z):=\Pr_I[\csim\leftarrow z\mid C_0,C_2]\cdot \E_{I}\left[\cP(Q,I)\mid \csim\leftarrow z,C_0,C_2\right].
	    \]
	    If the lemma premises are satisfied and $k\leq 4t_{\total}$, let $Y\subset \csim$ be a subset of cells of size $ \kappa := |Y| \leq O\left(\sqrt{t_{\total} (t_q\lg 1/p+\lg 1/\beta)}\right)$, which is guaranteed to exist by the lemma.
	\end{enumerate}

	\vspace{0.3ex}
	\end{minipage}\\
	\hline
	\end{tabular}
	\caption{The one-way weak simulation protocol of data structure $D$. }\label{figure:pi}
	\end{figure}

	\paragraph{Communication cost.}
	By the description of the algorithm, Alice only sends $C_0$ if it has less than $2pS$ cells, and only sends $C_1$ if it has less than $2pS$ \emph{updated} cells.
	Thus, sending $C_0$ and $C_1$ takes at most $(4pS-2)\cdot 2w$ bits (since sending the address or content takes $w$ bits).
	On the other hand, sending $C_2$ takes at most $S_{\cache}\cdot w$ bits. 
	Thus, the total communication takes at most
	\[
		(4pS-2)\cdot 2w+2+S_{\cache}\cdot w\leq (8pS+S_{\cache})\cdot w
	\]
	bits as claimed.

	We also note that by Chernoff bound, the probability that $\left|C_0\right|\geq 2pS$ is at most $2^{-\Omega(pS)}$, and the probability that $C_1$ has at least $2pS$ updated cells is also at most $2^{-\Omega(pS)}$.
	Since $\left|\cP(Q, I)\right|\leq 1$, these two events could only reduce the advantage by at most $2^{-\Omega(pS)}$.
	In the following, we will assume that Alice always sends $C_0$ and the updated cells in $C_1$ \emph{regardless of their sizes}, and prove the advantage of the protocol is at least
	\[
		2^{-O(\sqrt{t_{\total}(t_q\log 1/p+\log 1/\beta)}\cdot \log1/p)},
	\]
	which would imply the theorem.

	\paragraph*{Analysis of the advantage.}
	Consider the sampled cells $C_0$, for each query $Q$, let $\cW_Q$ denote the event that $C_0$ contains \emph{all updated cells that $Q$ probes}.
	Then we know that 
	\[
		\Pr_{C_0, I}\left[\cW_Q\right]\geq p^{4t_q}\cdot \Pr_{I}\left[t_q(Q,I)\leq 4t_q\right]\geq p^{4t_q}/4.
	\]
	If Bob knew whether $\cW_Q$ happens (i.e., if $\cW_Q$ was a function of $Q$ and Alice's message $C_0,C_2$), then we would be done:
	When $\cW_Q$ happens, Bob has enough information to simulate the query algorithm on query $Q$, thus, Bob knows the query output, i.e., the \emph{sign} of $\cP(Q, I)$.
	That is, $\cP(Q, I)$ has a fixed sign conditioned on $(C_0, C_2)$ and $Q$ for which $\cW_Q$ happens, and we have
	\[
		\left|\E_I[\cP(Q, I)\mid C_0,C_2]\right|=\E_I[\left|\cP(Q, I)\right|\mid C_0,C_2]\geq \beta.
	\]
	In this case, we would have
	\[
		\E_{Q, C_0,C_2}\left[\left|\E_I[\cP(Q, I)\mid C_0,C_2]\right|\right]\geq \Pr[\cW_Q]\cdot \E_{Q, C_0,C_2 \mid \cW_Q}\left[\left|\E_I[\cP(Q, I)\mid C_0,C_2]\right|\right]\geq \Omega(\beta\cdot p^{4t_q}),
	\]
	where the first inequality uses the assumption that $\cW_Q$ is determined by $Q$ and $C_0,C_2$.
	This is a much better advantage than what we claim.

	However, Bob \emph{does not} have enough information to determine whether all updated cells probed by $Q$ are sampled in $C_0$.
	This is because for the cells that are not in $C_0$ (or $C_2$), Bob cannot tell if they are updated but not sampled, or they are not updated and still have the pre-initialized contents.
	What Bob \emph{can} do is to generate the pre-initialized memory and update all cells in $C_0$ and $C_2$, then simulate the query algorithm on $Q$ on this memory state (see Step~\ref{step_bob} in Figure~\ref{figure:pi}). 
	This simulation can be incorrect, as Bob may not have the correct contents of the necessary cells.
	The simulation may even probe incorrect cells, as the query algorithm may be adaptive. 
	
	Nevertheless, let $\csim$ be the set of cells probed in this simulation.
	Now fix $\addr(\csim)$, the address of $\csim$, and let us consider the posterior distribution of the \emph{true} contents of $\csim$ in Bob's view, i.e., $\cont(\csim)\mid C_0,C_2$.\footnote{We use $\addr(C)$ to denote the addresses of cells $C$, and $\cont(C)$ to denote the contents of cells $C$.}
	Since $\cW_Q$ does happen with a nontrivial probability, there is a nontrivial probability that the simulation is in fact correct.
	Let $z^*$ be the content of $\csim$ that Bob uses for the simulation.
	This means that \emph{on average}, $z^*$ is slightly more likely than most other contents in this posterior distribution.
	Moreover, when the true content of $\csim$ is $z^*$, the expectation of $\cP(Q, I)$ conditioned on $\csim\leftarrow z^*$ is at least $\beta$ in the absolute value (the sign of $\cP(Q, I)$ is determined in this case).
	We define function $f$, which maps possible contents of $\csim$ to $\mathbb{R}$, to be the expectation of $\cP(Q, I)$ (i.e., the advantage) conditioned on the content, weighted by the probability in the posterior distribution (Step~\ref{step_pta}).
	Then $f$ takes a higher-than-average (absolute) value at $z^*$.
	The Peak-to-Average Lemma (Lemma~\ref{lem_peak_to_avg}) ensures that there is a \emph{small} subset $Y\subseteq \csim$ such that knowing the contents of $Y$ gives a nontrivial advantage.
	We further let Alice send another random set $C_1$ (Step~\ref{step_c1}).
	The final lower bound on the advantage comes from the event that $Y$ is contained in $C_1$.
	In this case, Bob can identify this event, since he can compute $Y$.
	Furthermore, it happens with a nontrivial probability since $Y$ is small, and once it happens, the advantage is nontrivial by Lemma~\ref{lem_peak_to_avg}.
	The formal argument goes as follows.

	\bigskip

	Let us fix $Q$ and consider a random $I$.
	The assumption on the query time of $Q$ (as we only consider $Q\in\cQ'$) guarantees that
	\[
		\E_{I}\left[t_{\total}(Q, I)\right]\leq t_{\total} \qquad\textrm{and} \qquad \E_{I}\left[t_{q}(Q, I)\right]\leq t_{q}.
	\]
	By Markov's inequality and union bound, we have $t_{\total}(Q, I)\leq 4t_{\total}$ and $t_q(Q, I)\leq 4t_q$ with probability at least $1/2$ (over the randomness of $I$).
	In this case, the probability of $\cW_Q$ ($C_0$ contains all updated cells that $Q$ probes) is
	\[
		\Pr_{C_0, I}\left[\cW_Q\mid t_{\total}(Q, I)\leq 4t_{\total},t_q(Q, I)\leq 4t_q\right] \geq p^{4t_q}.
	\]
	Therefore, we have
	\begin{equation}\label{eqn_avg_good}
		\Pr_{C_0, I}\left[\cW_Q\wedge t_{\total}(Q, I)\leq 4t_{\total}\right] \geq\Pr_{C_0, I}\left[\cW_Q\wedge t_{\total}(Q, I)\leq 4t_{\total}\wedge t_q(Q, I)\leq 4t_q\right] \geq p^{4t_q}/2.
	\end{equation}

	Now we say a pair $(C_0,C_2)$ is \emph{good} for query $Q$, if the simulation gives $\left|\csim\right|\leq 4t_{\total}$, and the event $\cW_Q$ is ``not-too-unlikely'' to happen conditioned on the pair: $\Pr_{I}\left[\cW_Q\mid C_0,C_2\right]\geq p^{4t_q}/4$.
	By~\eqref{eqn_avg_good}, Markov's inequality and the fact that $C_2$ is determined by $I$, we have
	\[
		\Pr_{C_0, C_2}\left[\Pr_{I}\left[\cW_Q\wedge t_{\total}(Q, I)\leq 4t_{\total}\mid C_0,C_2\right]\geq p^{4t_q}/4\right]\geq p^{4t_q}/4.
	\]
	This implies that $(C_0,C_2)$ is good for $Q$ with probability at least $p^{4t_q}/4$, since $\cW_Q$ implies that the simulation is correct, and we must have $\left|\csim\right|=t_{\total}(Q, I)$.

	In the following, let us only focus \emph{good} $(C_0, C_2)$ pairs. 
	Consider the \emph{posterior distribution} over the contents of $\csim$ \emph{conditioned on} a good pair $(C_0, C_2)$. 
	Let $z$ be a possible content of $\csim$.
	Define $f(z)$ to be the probability of $z$, multiplied by the expected value of $\cP(Q, I)$ conditioned on $z$ in this distribution (Step~\ref{step_pta}),
	\[
		f(z):=\Pr_I[\csim\leftarrow z\mid C_0,C_2]\cdot \E_{I}\left[\cP(Q, I)\mid \csim\leftarrow z, C_0, C_2\right].
	\]
	Intuitively, $f(z)$ is the \emph{contribution} to $\E[\cP(Q, I)\mid C_0,C_2]$ when $\csim\leftarrow z$.
	We have following claim about good pairs.

	\begin{claim}
		When $(C_0,C_2)$ is a \emph{good} pair for $Q$, the lemma premises are satisfied and $k\leq 4t_{\total}$ in Step~\ref{step_pta}.
	\end{claim}

	To see this, by the definition of a good pair, we have $\left|\csim\right|\leq 4t_{\total}$, i.e., $k\leq 4t_{\total}$.
	For the premises of Lemma~\ref{lem_peak_to_avg}, clearly, $f$ is a function mapping $\Sigma^k$ to $\mathbb{R}$.
	The first condition on $f$ is always satisfied: Since $\left|\cP(Q, I)\right|\leq 1$, we have
	\begin{align*}
		\sum_{z\in\Sigma^k} \left|f(z)\right|&=\left|\sum_{z\in\Sigma^k} \Pr[\csim\leftarrow z\mid C_0,C_2]\cdot \E_I[\cP(Q,I)\mid \csim\leftarrow z,C_0,C_2]\right| \\
		&\leq \sum_{z\in\Sigma^k} \Pr[\csim\leftarrow z\mid C_0,C_2]\cdot \E_I[\left|\cP(Q,I)\right|\mid \csim\leftarrow z,C_0,C_2] \\
		&\leq \sum_{z\in\Sigma^k} \Pr[\csim\leftarrow z\mid C_0,C_2] \\
		&=1.
	\end{align*}
	For the second condition on $f$, since $(C_0, C_2)$ is good, $\Pr_{I}\left[\cW_Q\mid C_0,C_2\right]\geq p^{4t_q}/4=\varepsilon/\beta$.
	When $\cW_Q$ happens, Bob will use the correct contents of $\csim$ for the simulation.
	Let $z^*$ be the content Bob uses in Step~\ref{step_bob}.
	Hence, a good pair implies $\Pr_I[\csim\leftarrow z^*\mid C_0,C_2]\geq \varepsilon/\beta$.
	Moreover, conditioned on $\csim\leftarrow z^*$ and $(C_0, C_2)$, the query output of $Q$ is determined (since the query algorithm is correctly simulated).
	That is, the \emph{sign} of $\cP(Q, I)$ is determined.
	Therefore, we have
	\begin{align*}
		\left|f(z^*)\right|&=\Pr_{I}\left[\csim\leftarrow z^*\mid C_0,C_2\right]\cdot \left|\E_{I}\left[\cP(Q, I)\mid \csim \leftarrow z^*, C_0, C_2\right]\right|\\
		&=\Pr_{I}\left[\csim\leftarrow z^*\mid C_0,C_2\right]\cdot \E_{I}\left[\left|\cP(Q, I)\right|\mid \csim\leftarrow z^*, C_0, C_2\right]\\
		&\geq \varepsilon,
	\end{align*}
	where the last inequality uses the fact that $\left|\cP(Q, I)\right|\geq \beta$ and $\Pr_I[\csim\leftarrow z^*\mid C_0,C_2]\geq \varepsilon/\beta$.

	Thus, Lemma~\ref{lem_peak_to_avg} guarantees that there exists a set $Y\subseteq \csim$ of size $$\kappa\leq O(\sqrt{k\cdot\log 1/\varepsilon})\leq O(\sqrt{t_{\total}(t_q\log 1/p+\log 1/\beta)}),$$ such that
	\begin{equation}\label{eqn_Y}
		\sum_{y\in \Sigma^Y}\left|\sum_{z\mid_Y=y}f(z)\right|\geq 2^{-O(\sqrt{t_{\total}(t_q\log 1/p+\log 1/\beta)})}.
	\end{equation}
	Note that Bob knows the set $Y$, since it is determined by the pair $(C_0, C_2)$ (and $Q$).
	Note that the LHS of the inequality is the expected advantage conditioned only on $Y$, i.e., 
	\begin{align}
		\sum_{y\in \Sigma^Y}\left|\sum_{z\mid_Y=y}f(z)\right|&=\sum_{y\in \Sigma^Y}\left|\sum_{z\mid_Y=y}\Pr\left[\csim\leftarrow z\mid C_0, C_2\right]\cdot \E_{I}\left[\cP(Q, I)\mid \csim\leftarrow z, C_0, C_2\right]\right| \nonumber\\
		&=\sum_{y\in \Sigma^Y}\left|\Pr[Y\leftarrow y\mid C_0, C_2]\cdot\sum_{z\mid_Y=y}\Pr[\csim\leftarrow z\mid Y\leftarrow y,C_0,C_2] \cdot \E_I\left[\cP(Q, I)\mid \csim\leftarrow z, C_0, C_2\right]\right|. \nonumber \\
		&=\sum_{y\in \Sigma^Y}\left|\Pr[Y\leftarrow y\mid C_0, C_2]\cdot \E_I\left[\cP(Q, I)\mid Y\leftarrow y, C_0, C_2\right]\right|. \label{eqn_Y2}
	\end{align}

	Finally, Alice further samples every cell with probability $p$ using public randomness, obtains $C_1$, and sends all updated cells to Bob (Step~\ref{step_c1}).
	Then with probability $p^{\kappa}$, we have $C_1\supseteq Y$.
	Note that Bob gets to know the contents of \emph{all} cells in $C_1$: If a cell is not updated, then it has the pre-initialized content.
	Furthermore, for a good pair $(C_0, C_2)$ and $\addr(C_1)$ such that $C_1\supseteq Y$, the expected advantage is
	\begin{align*}
		&\quad\,\E_{\cont(C_1)\mid C_0, C_2}\left[\left|\E_{I}\left[\cP(Q, I)\mid C_0,C_1,C_2\right]\right|\right] \\
		&\geq \E_{\cont(Y)\mid C_0, C_2}\left[\left|\E_{I}\left[\cP(Q, I)\mid C_0,\cont(Y),C_2\right]\right|\right] \\
		&=\sum_{y\in \Sigma^Y}\Pr[Y\leftarrow y\mid C_0, C_2]\cdot \left|\E_{I}\left[\cP(Q, I)\mid C_0,C_2,Y\leftarrow y\right]\right| \\
		&\geq 2^{-O(\sqrt{t_{\total}(t_q\log 1/p+\log 1/\beta)})},
	\end{align*}
	where the last inequality is by Equation~\eqref{eqn_Y} and~\eqref{eqn_Y2}.

	To prove that $\pi$ has the claimed advantage, we have for every $Q\in \cQ'$,
	\begin{align*}
		&\quad\,\,\E_{\msg}\left[\left|\E_I\left[\cP(Q, I)\mid \msg\right]\right|\right] \\
		&\geq \Pr\left[(C_0,C_2)\textrm{ is good}\wedge C_1\supseteq Y\right]\cdot \E_{\msg}\left[\left|\E_I\left[\cP(Q, I)\mid \msg\right]\right|\mid (C_0,C_2)\textrm{ is good}\wedge C_1\supseteq Y\right] \\
		&\geq \Omega(p^{4t_q}\cdot p^{\kappa})\cdot \E_{C_0,C_2,\addr(C_1)}\left[\E_{\cont(C_1)\mid C_0,C_2}\left[\left|\E_I\left[\cP(Q, I)\mid C_0,C_1,C_2\right]\right|\right]\mid (C_0,C_2)\textrm{ is good}\wedge C_1\supseteq Y\right] \\
		&\geq \Omega(p^{4t_q}\cdot p^{\kappa})\cdot2^{-O(\sqrt{t_{\total}(t_q\log 1/p+\log 1/\beta)})} \\
		&=2^{-O(\sqrt{t_{\total}(t_q\log 1/p+\log 1/\beta)}\cdot \log1/p)}.
	\end{align*}
	This proves the theorem.
\end{proof}

\bibliographystyle{abbrv}
\bibliography{refs}

\begin{thebibliography}{10}

\bibitem{apsp}
A.~Abboud, F.~Grandoni, and V.~Vassilevska~Williams.
\newblock Subcubic equivalences between graph centrality problems, apsp, and
  diameter.
\newblock {\em ACM Trans. Algorithms}, 19(1), mar 2023.

\bibitem{ssspneg}
A.~Bernstein, D.~Nanongkai, and C.~Wulff{-}Nilsen.
\newblock Negative-weight single-source shortest paths in near-linear time.
\newblock In {\em {FOCS}}, pages 600--611. {IEEE}, 2022.

\bibitem{maxflow}
L.~Chen, R.~Kyng, Y.~P. Liu, R.~Peng, M.~P. Gutenberg, and S.~Sachdeva.
\newblock Maximum flow and minimum-cost flow in almost-linear time.
\newblock In {\em {FOCS}}, pages 612--623. {IEEE}, 2022.

\bibitem{CGL15}
R.~Clifford, A.~Gr{\o}nlund, and K.~G. Larsen.
\newblock New unconditional hardness results for dynamic and online problems.
\newblock In {\em Proc. 56th IEEE Symposium on Foundations of Computer
  Science}, 2015.

\bibitem{FS89}
M.~L. Fredman and M.~E. Saks.
\newblock The cell probe complexity of dynamic data structures.
\newblock In {\em Proceedings of the 21st Annual {ACM} Symposium on Theory of
  Computing}, pages 345--354, 1989.

\bibitem{OMV}
M.~Henzinger, S.~Krinninger, D.~Nanongkai, and T.~Saranurak.
\newblock Unifying and strengthening hardness for dynamic problems via the
  online matrix-vector multiplication conjecture.
\newblock In {\em {STOC}}, pages 21--30. {ACM}, 2015.

\bibitem{seth}
R.~Impagliazzo, R.~Paturi, and F.~Zane.
\newblock Which problems have strongly exponential complexity?
\newblock {\em J. Comput. Syst. Sci.}, 63(4):512–530, dec 2001.

\bibitem{Larsen12a}
K.~G. Larsen.
\newblock The cell probe complexity of dynamic range counting.
\newblock In {\em Proceedings of the 44th Symposium on Theory of Computing
  Conference, {STOC} 2012}, pages 85--94, 2012.

\bibitem{Larsen:2012:focs}
K.~G. Larsen.
\newblock Higher cell probe lower bounds for evaluating polynomials.
\newblock In {\em Proc. 53rd IEEE Symposium on Foundations of Computer
  Science}, pages 293--301, 2012.

\bibitem{LWY18}
K.~G. Larsen, O.~Weinstein, and H.~Yu.
\newblock Crossing the logarithmic barrier for dynamic boolean data structure
  lower bounds.
\newblock {\em {SIAM} J. Comput.}, 49(5), 2020.

\bibitem{PTW10}
R.~Panigrahy, K.~Talwar, and U.~Wieder.
\newblock Lower bounds on near neighbor search via metric expansion.
\newblock In {\em 51th Annual {IEEE} Symposium on Foundations of Computer
  Science, {FOCS} 2010}, pages 805--814, 2010.

\bibitem{Pat11}
M.~P{\v a}tra\c{s}cu.
\newblock Unifying the landscape of cell-probe lower bounds.
\newblock {\em {SIAM} J. Comput.}, 40(3):827--847, 2011.

\bibitem{PD04a}
M.~P{\v a}tra\c{s}cu and E.~D. Demaine.
\newblock Tight bounds for the partial-sums problem.
\newblock In {\em Proceedings of the Fifteenth Annual {ACM-SIAM} Symposium on
  Discrete Algorithms, {SODA} 2004}, pages 20--29, 2004.

\bibitem{Patrascu07}
M.~P{\v a}tra{\c s}cu.
\newblock Lower bounds for 2-dimensional range counting.
\newblock In {\em Proc. 39th ACM Symposium on Theory of Computation}, pages
  40--46, 2007.

\bibitem{multiphase}
M.~P{\v a}tra{\c s}cu.
\newblock Towards polynomial lower bounds for dynamic problems.
\newblock In {\em Proceedings of the 42nd {ACM} Symposium on Theory of
  Computing, {STOC} 2010, Cambridge, Massachusetts, USA, 5-8 June 2010}, pages
  603--610. {ACM}, 2010.

\bibitem{PT06}
M.~P{\v a}tra{\c s}cu and M.~Thorup.
\newblock Higher lower bounds for near-neighbor and further rich problems.
\newblock In {\em 47th Annual {IEEE} Symposium on Foundations of Computer
  Science {FOCS} 2006}, pages 646--654, 2006.

\bibitem{PT11}
M.~P\v{a}tra\c{s}cu and M.~Thorup.
\newblock Don't rush into a union: take time to find your roots.
\newblock In {\em Proceedings of the 43rd {ACM} Symposium on Theory of
  Computing, {STOC} 2011}, pages 559--568, 2011.

\bibitem{siegel}
A.~Siegel.
\newblock On universal classes of extremely random constant-time hash
  functions.
\newblock {\em SIAM Journal on Computing}, 33(3):505--543, 2004.

\bibitem{Yao81}
A.~C. Yao.
\newblock Should tables be sorted?
\newblock {\em J. {ACM}}, 28(3):615--628, 1981.

\end{thebibliography}

\end{document}